\documentclass{lmcs}
\pdfoutput=1

\usepackage{lastpage}
\lmcsdoi{22}{1}{18}
\lmcsheading{}{\pageref{LastPage}}{}{}%
{Jan.~16,~2025}{Mar.~10,~2026}{}

\usepackage{xargs}
\usepackage[utf8]{inputenc}
\usepackage{graphicx}
\usepackage{wrapfig,booktabs}
\usepackage{xspace}
\usepackage{subcaption}
\usepackage[pdfpagelabels]{hyperref}
\usepackage{verbatimbox}
\usepackage{microtype}
\usepackage{listings}
\usepackage{textcomp}
\usepackage{tikzmacros}
\usetikzlibrary{positioning,arrows,matrix}
\usepackage[noend]{algpseudocode}
\usepackage{pdfpages}
\usepackage{mathtools}
\newcommand{\gapi}{Graph API\xspace}
\newcommand{\ie}{, i.e,\xspace}
\newcommand{\eg}{, e.g.,\xspace}
\definecolor{mauve}{rgb}{0.58,0,0.82}
\definecolor{whitesmoke}{rgb}{0.96, 0.96, 0.96}

\lstdefinelanguage{myLang}{
	sensitive = true,
	alsoletter=-,
	keywords={type},
	otherkeywords={},
	morestring=[b]",
	tabsize=4}
\newcommand{\lstMyLang}[1]{\lstinline[breaklines=true,language=myLang,basicstyle=\listingsfontinline,mathescape,literate={\-}{}{0\discretionary{-}}}
\lstdefinestyle{mystyle}{
	backgroundcolor=\color{whitesmoke},
	keywordstyle=\color{mauve},
	stringstyle=\ttfamily\tiny,
	basicstyle=\ttfamily\tiny,
	breakatwhitespace=false,
	breaklines=true,
	captionpos=b,
	keepspaces=true,
	numbers=left,
	numbersep=5pt,
	showspaces=false,
	showstringspaces=false,
	showtabs=false,
	tabsize=2
}
\usepackage{ifthen}
\usepackage{amssymb}
\AtBeginDocument{}%
\AtBeginDocument{}%
\begin{document}

\title[Taint Analysis for Graph APIs Focusing on Broken Access Control]{Taint Analysis for Graph APIs\texorpdfstring{\\}{} Focusing on Broken Access Control}
\author[L.~Lambers]{Leen Lambers\lmcsorcid{0000-0001-6937-5167}}[a]
\author[L.~Sakizloglou]{Lucas Sakizloglou\lmcsorcid{0000-0001-6971-1589}}[a]
\author[T.~Khakharova]{Taisiya Khakharova\lmcsorcid{0009-0006-7720-1160}}[a]
\author[F.~Orejas]{Fernando Orejas\lmcsorcid{0000-0002-3023-4006}}[b]

\address{Brandenburg University of Technology, Cottbus, Germany}
\email{<name>.<surname>@b-tu.de}
\address{Universitat Politècnica de Catalunya, Barcelona, Spain}
\email{orejas@lsi.upc.edu}

\begin{abstract}
A Graph API models the data managed by a web application based on the graph data model, i.e., data objects are represented by nodes and relationships by edges. This rather novel type of APIs presents new challenges when it comes to properly securing the APIs against the usual web application security risks, e.g., broken access control. A prominent security testing approach is taint analysis, which traces tainted, i.e., security-relevant, data from sources (where tainted data is inserted) to sinks (where the use of tainted data may lead to a security risk), over the information flow in an application.

We present the first systematic approach to \emph{static} and \emph{dynamic} taint analysis for Graph APIs focusing on broken access control. The approach comprises the following. We taint nodes of the Graph API if they represent data requiring specific privileges in order to be retrieved or manipulated, and identify API calls which are related to sources and sinks. Then, we statically analyze whether a tainted information flow between API source and sink calls occurs. To this end, we model the API calls using graph transformation rules. We subsequently use Critical Pair Analysis to automatically analyze potential dependencies between rules representing source calls and rules representing sink calls.
We distinguish direct from indirect tainted information flow and argue under which conditions the Critical Pair Analysis is able to detect not only direct, but also indirect tainted flow.
The static taint analysis (i) identifies flows that need to be further reviewed, since tainted nodes may be created by an API call and used or manipulated by another API call later without having the necessary privileges, and (ii) can be used to systematically design dynamic security tests for broken access control. The dynamic taint analysis checks if potential broken access control risks detected during the static taint analysis really occur. We apply the approach to a part of the GitHub GraphQL API. The application illustrates that our analysis supports the detection of two types of broken access control systematically: the case where users of the API may not be able to access or manipulate information, although they should be able to do so; and the case where users (or attackers) of the API may be able to access/manipulate information that they should not.%
\end{abstract}
\maketitle
\section{Introduction}\label{sec:intro}
A \emph{Graph API} models the data managed by a web application based on the graph data model, i.e., data objects are represented by nodes and relationships by edges.
While initially they emerged to facilitate data management in social networks, Graph APIs are currently used by an increasing number of software applications in the area of service-oriented and cloud computing~\cite{Zhang23,Quina-MeraFGR23}.
A prominent specification and query language for Graph APIs is \emph{GraphQL}~\cite{GraphQLFoundation__GraphQLquerylanguageyourAPI}.
GraphQL was introduced by Facebook in 2012 as part of the development of mobile applications, and is now used in various web and mobile applications, e.g., GitHub, Pinterest, Paypal, etc.

\emph{Security}~\cite{Graw07} is particularly important in the area of web and mobile applications.
Graph APIs are still relatively new, which means that a wide range of methods and tools for security analyses are not yet available for them~\cite{Quina-MeraFGR23,Pagey+23}.
\emph{Taint analysis} is a typical representative of such security analyses.
It is a particular instance of data flow analysis~\cite{PezzeYoung07}
and systematically tracks certain user inputs through the application.
\emph{Broken access control} is considered as one of the most prominent security vulnerabilities for web applications~\cite{TeamOWASPTop10__A01BrokenAccessControlOWASPTop102021}.
Access control enforces policies such that users cannot act outside their intended permissions. Failures typically lead to unauthorized information disclosure, modification, or destruction of all data or performing a business function outside the user's limits. Typical access control failures occur when an API can be accessed with (i) missing or misunderstood access control, or (ii) wrongly implemented access control.
Such failures occur in practice, as demonstrated by issues in the GitHub community forum\eg issues 110618, 106598, 85661~\cite{GitHub__Discussions}.

In this paper, we present a \emph{taint analysis for Graph APIs} focusing on \emph{broken access control} (BAC).
Specifically, our approach builds on the concept of taint analysis, which propagates taint labels through the implementation of a program, to trace the flow of sensitive data through the Graph API.
Graph transformation is a technique dedicated to formalizing graph manipulations and is therefore a natural foundation for formalizing a \gapi.
Specifically, we formalize both a \gapi and the taint analysis using concepts from graph transformation~\cite{1997handbook,Ehrig_2006_FundamentalsAlgebraicGraphTransformation,HeckelT20} and rely on related, formally founded tool support to perform part of the analysis.
A \emph{security analyst} can use the taint analysis to systematically and soundly check for BAC vulnerabilities.
To the best of our knowledge, no other approach to checking for BAC vulnerabilities simultaneously supports key feature of our approach: use of taint analysis, native graph formalization, and sound detection---see discussion in \autoref{sec:related}.
We illustrate the taint analysis based on GraphQL, but our approach is applicable to any kind of Graph APIs that can be formalized using graph transformation.

The presented taint analysis comprises three segments: the setup, the static analysis, and the dynamic analysis. During \emph{setup}, the security analyst can taint nodes of the Graph API that represent sensitive data.
This Graph API is represented formally using graph transformation rules typed over a type graph.
The security analyst thus taints nodes in the type graph, representing node types entailing sensitive data that may be processed by the Graph API.
From these tainted nodes, we can characterize graph transformation rules as Graph API \emph{sources}, where tainted data are inserted\ie created by the rule, and \emph{sinks}, where tainted data are used\ie changed by the rule, and may thus lead to a security risk.
The \emph{static analysis} supports the security analyst in \emph{validating} the access control policies, i.e. finding missing or misunderstood access control policies.
This is achieved by automatically generating pairs of source and sinks that may be risky in terms of leaking sensitive data and thus represent potential BAC vulnerabilities.
The security analyst can then systematically check if the policy description indicates that it properly secures these risky flows.
Subsequently, the \emph{dynamic analysis} supports the security analyst in \emph{verifying} that the access control policies are implemented correctly.
If the security analyst expects a risky flow based on the results of the static analysis, then they can systematically try to uncover a corresponding BAC vulnerability in the Graph API implementation.
We demonstrate that the static analysis is \emph{sound} w.r.t. finding potential \emph{direct} BAC vulnerabilities\ie when a sink directly follows a source, as well as a type of \emph{indirect} ones\ie when other data manipulations occur between source and sink. However, since the static analysis is incomplete, we integrate it with a subsequent \emph{complete} dynamic analysis.

Graph transformation is a formal paradigm with many applications.
A \emph{graph transformation system} is a collection of graph transformation rules that, in union, serve a common purpose.
In this paper, a \emph{graph transformation system} is used to formally describe a Graph API.
For many applications (see \cite{Lambers_2018_Multigranularconflictdependencyanalysissoftwareengineeringbasedgraphtransformationa} for a survey involving 25 papers), it is beneficial to know all conflicts and dependencies that can occur for a given pair of rules.
A conflict is a situation in which one rule application renders another rule application inapplicable.
A dependency is a situation in which one rule application needs to be performed such that another rule application becomes possible.
We will use \emph{dependency analysis} to discover possible tainted information flows that may lead to broken access control vulnerabilities.

This article is an extension of a paper published in ICGT '24~\cite{Lambers_2024_TaintAnalysisGraphAPIsFocusingBrokenAccessControl}. The article introduces the following extensions. First, it extends the formalization to support a certain type of indirect vulnerabilities (see \autoref{thm:general}), whereas the paper supported only direct vulnerabilities. Moreover, we have added a number of insightful examples illustrating the limitation of the paper formalization as well as the new formal results. As a second extension, the article introduces a new type of coverage, namely \emph{role coverage} (see \autoref{def:role-coverage}), which addresses whether derived tests cover the various roles in an access control policy. Finally, the article introduces a new application of the presented taint analysis---see \autoref{subsec:graphql_real-issue}. This application reproduces the issue 106598 reported in the GitHub community forum, addresses a larger part of the GraphQL schema language, and increases the level of real-world detail covered by the examples. Additionally, we have made the following improvements and changes: we elaborated on preliminaries regarding the dependency analysis---see \autoref{subsec:preliminaries_graph-transformation}; and we restructured the section describing the applications of the taint analysis to GitHub\ie\autoref{sec:graphql}, to render it more informative; we have added a summary of the GitHub access control policy, an outline of the application setting, as well as a discussion of the applications.

The rest of this paper is structured as follows:
Section~\ref{sec:preliminaries} revisits preliminaries.
Section~\ref{sec:analysis} presents taint analysis\ie the setup, the static, and the dynamic analysis, in further detail, based on a basic running example.
Section~\ref{sec:graphql} applies the analysis to a part of the GitHub GraphQL API~\cite{GitHub__GitHubGraphQLAPIdocumentation}.
Section~\ref{sec:related} is dedicated to related work while Section~\ref{sec:conclusion} to conclusion.
\section{Preliminaries}
\label{sec:preliminaries}
First, we review the basics of Graph APIs and GraphQL, a prominent Graph API discussed further in the paper; in the same section, we introduce the running example used in the remainder.
Then we present the basics of access control.
Finally, we recall the double-pushout (DPO) approach to graph transformation as presented in~\cite{Ehrig_2006_FundamentalsAlgebraicGraphTransformation}. We reconsider dependency notions on the transformation and rule level~\cite{Lambers_2008_EfficientConflictDetectionGraphTransformationSystemsEssentialCriticalPairs,Lambers_2018_InitialConflictsDependenciesCriticalPairsRevisited,Lambers_2018_Multigranularconflictdependencyanalysissoftwareengineeringbasedgraphtransformationa} for the DPO approach, including dependency reasons.

\subsection{Graph APIs and Running Example}\label{subsec:preliminaries_graphAPIs}
A Graph API models the data managed by a web application based on the graph data model, i.e., data objects are represented by nodes and relationships by edges.
Graph APIs emerged to facilitate data management in social networks which similarly relied on graph-structured data. Two familiar examples of Graph APIs are the \emph{Facebook API}~\cite{Meta__OverviewGraphAPIDocumentation} and the \emph{GitHub GraphQL API}~\cite{GitHub__GitHubGraphQLAPIdocumentation} (henceforth referred to as GitHub API). The GitHub API is based on the \emph{GraphQL} specification and query language that describes the capabilities and requirements of graph data models for web applications~\cite{GraphQLFoundation__GraphQLquerylanguageyourAPI}. In the remainder, we focus on APIs which are based on GraphQL.
\begin{figure}[t]
	\centering

\begin{minipage}[t]{.85\textwidth}
\begin{lstlisting}[frame=tlrb, numbers=none]{Name}
type User {
	username: String!
	projects: [Project]!
	repos: [Repository]!
}

type Repository {
	id: ID!
	name: String!
	description: String!
	url: String!
	owner: Repository!
}

type Issue {
	id: ID!
	body: String!
	repo: Repository!
}

type Project {
	name: String!
	id: ID!
	body: String
}

type Query {
	getProject(projectId: ID!): Project
}

type Mutation {
	createRepo(newName: String!, newDescription: String!, repoURL: String!): ID
	updateRepo(repoId: ID!, newName: String!, newDescription: String!): Boolean
	createProject(projectId: ID!, body:String!, newName:String!): ID
	deleteProject(projectId: ID!): Boolean
	createIssue(repoId: ID!, newBody: String!): ID
	updateIssue(repoId: ID!, issueId, ID!, newBody: String!): Boolean
	deleteIssue(repoId: ID!, issueId: ID!): Boolean
}\end{lstlisting}
\end{minipage}
	\caption{\label{fig:schema}The GraphQL schema for the running example.}
\end{figure}

GraphQL uses a \emph{schema}, i.e., a strongly typed contract between the client and the server, to define all exposed types of an API as well as their structure. Schemata of GraphQL APIs are either publicly available or can be discovered using distinguished \emph{introspection} queries. An example of a schema is shown in \autoref{fig:schema}. The schema captures the GraphQL API of a simplified version of a platform that allows developers to manage their code as well as their development efforts, similar to GitHub. The schema is defined in the GraphQL schema language~\cite{GraphQLFoundation__SchemasTypesGraphQL} and contains the following object types: the \texttt{User}, which represents a user of the platform; the \texttt{Repository}, which represents a code repository hosted on the platform; the \texttt{Issue}, which represents a task related to the development work being performed in a repository; and the \texttt{Project}, which represents a task-board that can be used for planning and organizing development work. Attributes and relationships are captured by fields within types, where squared brackets refer to a list of objects. The exclamation mark renders a field non-nullable.

GraphQL defines two types of operations that can be performed on the server: \emph{queries} are used to fetch data, whereas \emph{mutations} are used to manipulate\ie create, update, or delete data. Queries and mutations are also defined in the schema, within the distinguished types \texttt{Query} and \texttt{Mutation}, respectively. \autoref{fig:schema} shows a minimal list of mutations for the platform in our example: the mutations \textit{createRepo}, \textit{updateRepo} create and update an instance of a \texttt{Repository}, respectively.%
\subsection{Access Control Policies}\label{subsec:preliminaries_access-control}

In various types of software systems, e.g., financial and safety-critical, the adequate security of the information and the systems themselves constitutes a fundamental requirement. \emph{Access control} refers to ensuring that users are allowed to perform activities which may access information or resources only when they have the required rights. An access control \emph{policy} is a collection of high-level requirements that manage user access and determine who may access information and resources under which circumstances~\cite{Hu_2017_VerificationTestMethodsAccessControlPoliciesModels}. Correspondingly, a policy is \emph{broken} when these requirements are not fulfilled.

While there are various methods of realizing policies, we focus on the \emph{role-based} method. In a role-based policy, a security analyst defines roles, i.e., security groups, whose access rights are defined according to roles held by users and groups in organizational functions; the analyst assigns a role to each user who, once authorized, can perform the activities permitted by said role. In the remainder, we assume that, besides the security analyst, users cannot perform administrative actions\ie cannot change their own role. A role-based access control policy is broken when a user performs an activity that is not permitted by their role.

A requirement of a role-based policy for our running example could state that \emph{a user may only update the information of a repository, if that user is the owner of the repository}. In practice, the \gapi of the platform should ensure that the query \textit{updateRepo} can only be executed by a user who is authorized to update the repository\ie who has the role \emph{owner}. An update of the repository information by a \emph{collaborator}\ie a user who does not belong to this group, would break the policy.

\subsection{Graph Transformation and Static Dependency Analysis}\label{subsec:preliminaries_graph-transformation}
Throughout this paper we always consider graphs (and graph morphisms) typed over some fixed type graph $\mathit{TG}$ via a typing morphism $type_G:G\rightarrow TG$ as presented in \cite{Ehrig_2006_FundamentalsAlgebraicGraphTransformation}.
A graph morphism $m:G\rightarrow H$ between two graphs consists of two mappings $m_V$ and $m_E$ between their nodes and edges being both type-preserving and structure-preserving w.r.t. source and target functions.
Note that in the main text we denote inclusions by $\hookrightarrow$ and all other morphisms by $\rightarrow$.

\emph{Graph transformation} is the rule-based modification of graphs. A {\em rule} consists of two graphs:
$L$ is the left-hand side (LHS) of the rule representing a pattern that has to be found to apply the rule.
After the rule application, a pattern equal to $R$,  the right-hand side (RHS), has been created.
The intersection $K$ is the graph part that is not changed; it is equal to  $L \cap R$ provided that the result is a graph again.
The graph part that is to be deleted is defined by $L \setminus (L \cap R)$, while $R \setminus (L \cap R)$ defines the graph part to be created.
We consider a graph transformation system just as a set of rules.

A \emph{direct graph transformation} $G \stackrel{m,r}{\Longrightarrow} H$ between two graphs $G$ and $H$ is defined by first finding a graph morphism $m$ of the LHS $L$ of rule $r$ into $G$ such that $m$ is injective, and second by constructing $H$ in two passes:
(1) build $D := G \setminus m(L \setminus K)$,
i.e., erase all graph elements that are to be deleted;
(2) construct $H := D \cup m\rq{}(R \setminus K)$.
The morphism $m\rq{}$ has to be chosen such that a new copy of all graph elements that are to be created is added.
It has been shown that $r$ is applicable at $m$ iff $m$ fulfills the {\em dangling condition}.
It is satisfied if all adjacent graph edges of a graph node to be deleted  are deleted as well, such that $D$ becomes a graph.
Injective matches are usually sufficient in applications and w.r.t. our work here, they allow explaining constructions with more ease than for general matches.
In categorical terms, a direct transformation step is defined using a so-called double pushout as in the following definition.
Thereby step (1) in the previous informal explanation is represented by the first pushout (PO1) and step (2) by the second pushout (PO2)~\cite{Ehrig_2006_FundamentalsAlgebraicGraphTransformation}.

\begin{defi}[(read/write) rule and transformation]\label{def:transformation}
	A {\em rule} $r$ is defined by $r = (L  \stackrel{le}{\hookleftarrow} K \stackrel{ri}{\hookrightarrow} R)$ 
	with $L, K,$ and $R$ being graphs connected by two graph inclusions and $K = L \cap R$.
	A rule $r$ is a \emph{read rule} if $le$ and $ri$ are identity morphisms.
	A rule $r$ is a \emph{write rule} if it is not a read rule.
	\begin{flushleft}
		\begin{minipage}{0.55\textwidth}

			\noindent {\em 	A {\em direct transformation} $G \stackrel{m,m',r}{\Longrightarrow} H$ applying $r$ to $G$ consists of two pushouts as depicted on the right.
				Rule $r$ is {\em applicable} and the injective morphism $m: L \rightarrow G$ ($m':R \rightarrow H$) is called a {\em match} (resp. co-match) if there exists a graph $D$ such that $(PO1)$ is a pushout.}
		\end{minipage}%
		\hspace{.05cm}
		\begin{minipage}{0.4\textwidth}
			\begin{center}
				\begin{tikzpicture}[show background rectangle]
					\fill (1,2) node[inner sep=1pt] (L) {$L$};
					\fill (3,2) node[inner sep=1pt] (K) {$K$};
					\fill (5,2) node[inner sep=1pt] (R) {$R$};
					\fill (1,0.5) node[inner sep=1pt] (G) {$G$};
					\fill (3,0.5) node[inner sep=1pt] (D) {$D$};
					\fill (5,0.5) node[inner sep=1pt] (H) {$H$};

					\fill (2,1.25) node {$(PO1)$};
					\fill (4,1.25) node {$(PO2)$};

					{\pgfsetarrowsend{latex}
						\draw (K) -> node[above,inner sep=1pt]{$\scriptstyle{}$} (L);
						\draw (K) -> node[above,inner sep=1pt]{$\scriptstyle{}$} (R);
						\draw (D) -> node[below,inner sep=1pt]{$\scriptstyle{}$}(G);
						\draw (D) -> node[below,inner sep=1pt]{$\scriptstyle{}$}(H);
						\draw (L) -> node[left,inner sep=3pt]{$\scriptstyle{m}$} (G);
						\draw (K) -> node[left,inner sep=1pt]{$\scriptstyle{}$}(D);
						\draw (R) -> node[right,inner sep=1pt]{$\scriptstyle{m\rq{}}$}(H);
					}
				\end{tikzpicture}
			\end{center}
		\end{minipage}
	\end{flushleft}
\end{defi}

We may notice that a read rule is not really a transformation rule, in the sense that if we apply it to a graph $G$, the result is again $G$, i.e. it causes no transformation on it. We use read rules to model operations that just read data on a graph, without transforming it. In particular, when  applying the rule $r = (L  \stackrel{le}{\hookleftarrow} L \stackrel{ri}{\hookrightarrow} L)$  to $G$ via a morphism $m$, we consider that some given user is reading the data included in $m(L)$.

Given a pair of transformations, a \emph{produce-use dependency}~\cite{Ehrig_2006_FundamentalsAlgebraicGraphTransformation} occurs if the co-match of the first transformation via rule $r_1$ produces a match for the second transformation that was not available yet before applying the first transformation.  More precisely, the \emph{creation graph} $C_1$ includes all the the elements (nodes and edges) that are produced by $r_1$. In particular, $C_1$ is the smallest subgraph of $R_1$ that includes $R_1 \setminus K_1$. Notice that, in general, $R_1 \setminus K_1$ is not a graph, because it may include edges whose source or target is in $K_1$, i.e., $C_1$ would include, in addition to the edges in $R_1 \setminus K_1$, their source and target. These additional nodes are called boundary nodes and summarized in the \emph{boundary graph} $B_1$, i.e., $B_1 = C_1 \cap K_1$. Then, the  \emph{dependency reason} $C_1 \stackrel{o_1}{\hookleftarrow}  S_1 \stackrel{q_{12}}\rightarrow L_2$ comprises elements produced by the first and used by the second rule giving rise to the dependency. The formal construction of this dependency reason is depicted in \autoref{fig:dependencyChar}, where technically square (1) is an initial pushout and square (2) is the pullback of $m'_1 \circ c_1$ and $m_2$, i.e. $S_1$ contains all elements from $C_1$ and $L_2$ that are mapped to the same element in $H_1$.
In \cite{Lambers_2019_Granularityconflictsdependenciesgraphtransformationsystemstwodimensionalapproach} it is explained in detail how to obtain, using \emph{static dependency analysis}, such dependency reasons for a pair of rules for which transformation pairs in produce-use dependency exist.
Each transformation pair in produce-use dependency comes with a unique dependency reason (completeness).

\begin{defi}[produce-use dependency, dependency graph]\label{def:dependency}
	Given a pair of direct transformations $(t_1, t_2) = (G \stackrel{m_1,m'_1,r_1}{\Longrightarrow} H_1, H_1 \stackrel{m_2,m'_2,r_2}{\Longrightarrow} H_2)$ applying rules $r_1: L_1 \stackrel{le_1}{\hookleftarrow} K_1 \stackrel{ri_1}{\hookrightarrow} R_1$ and $r_2: L_2 \stackrel{le_2}{\hookleftarrow} K_2 \stackrel{ri_2}{\hookrightarrow} R_2$ as depicted in Fig.~\ref{fig:dependencyChar}.
	Square (1) in Figure~\ref{fig:dependencyChar} can be constructed as the initial pushout over morphism $\mathit{ri}_1$.
	It yields the {\em boundary graph} $B_1$ and the {\em creation graph} $C_1$.
	The \emph{transformation pair}  $(t_1,t_2)$ is in \emph{produce-use dependency} if there does not exist a morphism $x_{21}: L_2 \rightarrow D_1$ such that $h_1 \circ x_{21} = m_2$.
	The \emph{dependency graph} $DG(\mathcal{R})$ for a set of rules $\mathcal{R}$ consists of a node for each rule in $\mathcal{R}$ and an edge $(r_{1},r_{2})$ from $r_1$ to $r_2$ if there exists a transformation pair $(t_1, t_2) = (G \stackrel{r_1}{\Longrightarrow} H_1, H_1 \stackrel{r_2}{\Longrightarrow} H_2)$ in produce-use dependency.
\end{defi}

\begin{figure}[t]

	\centering
	\begin{tikzpicture}[show background rectangle]
		\fill (1,2) node[inner sep=1pt] (L2) {$L_2$};
		\fill (3,2) node[inner sep=1pt] (K2) {$K_2$};
		\fill (5,2) node[inner sep=1pt] (R2) {$R_2$};
		\fill (0,0.5) node[inner sep=1pt] (H1) {$H_1$};
		\fill (3,0.5) node[inner sep=1pt] (D2) {$D_2$};
		\fill (5,0.5) node[inner sep=1pt] (H2) {$H_2$};

		{\pgfsetarrowsend{latex}
			\draw (K2) -> node[above,inner sep=1pt]{$\scriptstyle{le_2}$} (L2);
			\draw (K2) -> node[above,inner sep=1pt]{$\scriptstyle{ri_2}$} (R2);
			\draw (D2) -> node[below,inner sep=1pt]{$\scriptstyle{g_2}$}(H1);
			\draw (D2) -> node[below,inner sep=1pt]{$\scriptstyle{h_2}$}(H2);
			\draw (L2) -> node[right,inner sep=3pt,pos=.7]{$\scriptstyle{m_2}$} (H1);
			\draw (K2) -> node[left,inner sep=1pt]{$\scriptstyle{d_2}$}(D2);
			\draw (R2) -> node[right,inner sep=1pt]{$\scriptstyle{m\rq{}_2}$}(H2);
		}

		\fill (-2,2.75) node {$(1)$};
		\fill (0,2) node {$(2)$};

		\fill (0,4.5) node[inner sep=1pt] (A1) {$S_1$};
		\fill (-1,3.5) node[inner sep=1pt] (C1) {$C_1$};
		\fill (-3,3.5) node[inner sep=1pt] (B1) {$B_1$};
		\fill (-1,2) node[inner sep=1pt] (R1) {$R_1$};
		\fill (-3,2) node[inner sep=1pt] (K1) {$K_1$};
		\fill (-5,2) node[inner sep=1pt] (L1) {$L_1$};
		\fill (-3,0.5) node[inner sep=1pt] (D1) {$D_1$};
		\fill (-5,0.5) node[inner sep=1pt] (G) {$G$};

		{\pgfsetarrowsend{latex}
			\draw[dashed] (A1) -> node[above,inner sep=1pt]{$\scriptstyle{x}$} (B1);
			\draw (A1) -> node[above,inner sep=3pt]{$\scriptstyle{o_1}$} (C1);
			\draw (A1) -> node[above,inner sep=22pt]{$\scriptstyle{q_{12}}$} (L2);
			\draw (B1) -> node[above,inner sep=1pt]{$\scriptstyle{b_1}$} (C1);
			\draw (B1) -> node[left,inner sep=1pt]{$\scriptstyle{}$} (K1);
			\draw (C1) -> node[left,inner sep=1pt]{$\scriptstyle{c_1}$} (R1);

			\draw (K1) -> node[above,inner sep=1pt]{$\scriptstyle{ri_1}$} (R1);
			\draw (K1) -> node[above,inner sep=1pt]{$\scriptstyle{le_1}$} (L1);
			\draw (D1) -> node[below,inner sep=1pt]{$\scriptstyle{h_1}$}(H1);
			\draw (D1) -> node[below,inner sep=1pt]{$\scriptstyle{g_1}$}(G);
			\draw (R1) -> node[left,inner sep=3pt,pos=.7]{$\scriptstyle{m'_1}$} (H1);
			\draw[dashed] (L2) -> node[above,inner sep=1pt]{$\scriptstyle{x_{21}}$} (D1);
			\draw (K1) -> node[left,inner sep=1pt]{$\scriptstyle{d_1}$}(D1);
			\draw (L1) -> node[left,inner sep=1pt]{$\scriptstyle{m_1}$}(G);
		}
	\end{tikzpicture}
	\caption{Illustration of dependency and dependency reason.}

	\label{fig:dependencyChar}

\end{figure}

For a given graph transformation system, a \textit{static dependency analysis} technique is a means to compute a list of all pairwise produce-use dependencies.
Inspired by the related concept from term rewriting, \textit{critical pair analysis} (CPA, \cite{Plump_1994_Criticalpairstermgraphrewriting}) has been the established static conflict and dependency analysis technique for three decades.
CPA reports each dependency as a critical pair, that is, a minimal example graph together with a pair of rule applications from which a dependency arises.
The \textit{multi-granular static dependency analysis} \cite{Born_2017_GranularityConflictsDependenciesGraphTransformationSystems,Lambers_2018_InitialConflictsDependenciesCriticalPairsRevisited}) supports the computation of dependencies for graph transformation systems on three granularity levels:
on binary granularity, it reports if a given rule pair contains a dependency at all;
on coarse granularity, it reports \textit{minimal dependency reasons};
fine granularity is, roughly speaking, the level of granularity provided by critical pairs.
Coarse-grained results have been shown to be more usable than fine-grained ones in a diverse set of scenarios and can be used to compute the fine-grained results much faster~\cite{Lambers_2018_Multigranularconflictdependencyanalysissoftwareengineeringbasedgraphtransformationa}. We use the coarse granularity also in the following.
\section{Taint Analysis for Graph APIs}\label{sec:analysis}
This section presents our main contribution, the taint analysis for Graph APIs focusing on broken access control. The analysis assumes the following input: a \emph{\gapi schema} (or a \gapi endpoint from which the schema can be introspected---see \autoref{subsec:preliminaries_graphAPIs}), comprising object and edge types as well as API calls\ie queries and mutations, and a \emph{role-based access control policy} description (henceforth referred to as policy), consisting of high-level access control requirements like the one mentioned in \autoref{subsec:preliminaries_access-control}.
The analysis comprises three segments: the setup, the static analysis, and the dynamic analysis---see \autoref{fig:activity-diagram}. The \emph{setup} segment comprises the derivation of a \emph{tainted Graph API}\ie a \gapi \emph{schema with tainted nodes} as well as \emph{source and sink graph transformation rules} derived from the policy and the \gapi schema; we describe the setup in detail in \autoref{subsec:analysis_formalization-API}. Once the setup is complete, the \emph{static analysis} can be performed. The output of the static analysis is a set of \emph{rule dependencies} which indicate potential occurrences of a broken policy. We present the formal concepts of the static analysis in \autoref{subsec:analysis_static}. The completion of the static analysis enables the execution of the \emph{dynamic analysis}. The dynamic analysis relies on the output of the static analysis and itself generates a \emph{test case report} which either verifies implementations of policy rules or exposes rules where the policy is broken. We present the formal concepts of the dynamic analysis in \autoref{subsec:analysis_dynamic}.\\
\noindent 
In the following, illustrations of type graphs and graph transformation rules are mostly based on the well-known \emph{Henshin} project~\cite{Arendt_2010_HenshinAdvancedConceptsToolsInPlaceEMFModelTransformations}, which is also used for tool support by the static analysis.
\begin{figure}[t]
	\centering
	\includegraphics[scale=.65]{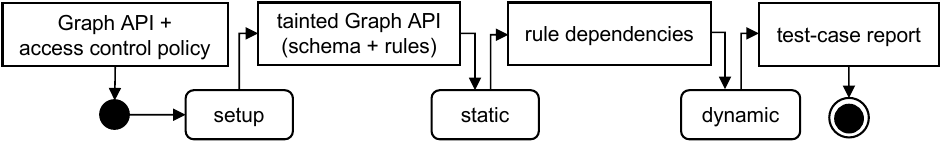}
	\caption{\label{fig:activity-diagram} The segments of the presented approach, illustrated as an activity diagram.}
\end{figure}
\subsection{Setup}\label{subsec:analysis_formalization-API}
As mentioned in \autoref{sec:intro}, the analysis relies on a formal description of a Graph API based on a graph transformation system, and focuses on broken access control. Therefore, our approach to taint analysis relies on the security analyst (i) obtaining both a representation of the GraphQL schema as a type graph $TG$ and a representation of the API\ie all possible API calls, as a set of graph transformation rules $\mathcal{R}$ typed over $TG$;  (ii) having access to a role-based policy specification for the Graph API; (iii) identifying nodes in the Graph API which represent sensitive data to be checked for broken access control by the taint analysis, as well as the rules in $\mathcal{R}$ which create and manipulate those data. We formally define these concepts below.
\begin{figure}[t]
	\centering
	\includegraphics[width=.75\textwidth]{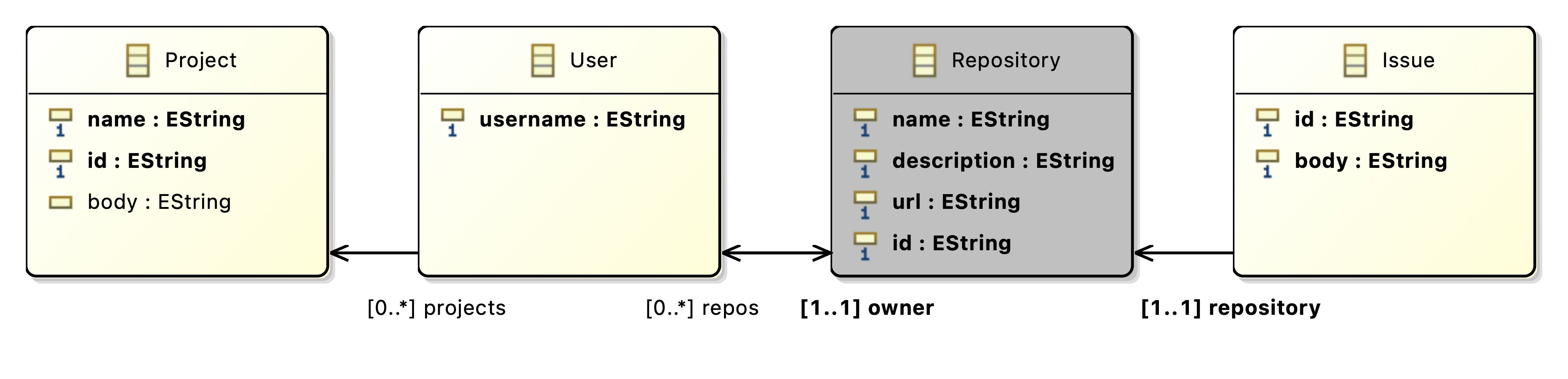}
	\caption{\label{fig:typeGraph} The tainted type graph, where \texttt{Repository} is a tainted node (in gray).}
\end{figure}
\subsubsection{Graph APIs}

We formally represent a Graph API by a set of graph transformation rules $\mathcal{R}$ typed over $TG$. We represent each API call as a direct graph transformation.
A sequence of API calls, denoted here as API execution, is a graph transformation of arbitrary length.
It starts with some initial graph $G_0$, representing the data present in the API prior to the execution.

\begin{defi}[graph API, call, execution]\label{def:graphAPI}
	A \emph{Graph API} with a schema $TG$ is a set of graph transformation rules $\mathcal{R} = \{r_i | i \in I\}$ typed over $TG$.
	A \emph{Graph API execution} consists of a (possibly infinite) graph transformation sequence $(t : G_0 \stackrel{r_1}{\Rightarrow} G_1 \stackrel{r_2}{\Rightarrow} G_2 \stackrel{r_3}{\Rightarrow} G_3 \ldots)$, starting with an
	initial graph $G_0$, such that each $r_i$ is an element of $\mathcal{R}$.
	A \emph{Graph API call} is a direct transformation within a Graph API execution.
	We denote with $Sem(\mathcal{R})$ the set of all Graph API executions, and with $Sem_{call}(\mathcal{R})$ the set of all Graph API calls.
\end{defi}

\noindent An illustration of a type graph for the platform of the running example is shown in \autoref{fig:typeGraph}. The type graph corresponds to the entity types in the schema defined in the GraphQL schema language in \autoref{fig:schema}.
Based on the mutations defined within the distinguished type \texttt{Mutation} in the GraphQL schema, we also illustrate a Graph API according to \autoref{def:graphAPI}\ie a set of transformation rules typed over the type graph, in \autoref{subfig:repo-rules}.
We show an example of a Graph API execution comprising three calls\ie rules, in \autoref{fig:reducibleVulnerability}.
In the remainder, we illustrate rules using two different notations: the \emph{plain} notation (seen in Figures~\ref{fig:reducibleVulnerability},~\ref{fig:deleteIncidentWhole},~\ref{fig:createIncidentWhole}, and~\ref{fig:reducedVulnerabilityWhole}) illustrates the LHS and RHS as separate graphs and omits certain information\eg attributes; the \emph{compact} notation (seen in Figures~\ref{fig:rules} and~\ref{fig:rules-real} and generated by Henshin) embeds the LHS and RHS of a rule in one graph and does not omit any information. Note that figures using the plain notation, show rules on top, while the execution\ie graphs before and after the application of a rule, are shown at the bottom. In the Graph API execution in \autoref{fig:reducibleVulnerability}, a user creates a repository that is updated after the user creates a project and applies the rules \emph{createRepo}, \emph{createProject} and \emph{updateRepo}.

In this section, we will use a minimal API consisting of the rules \emph{createRepo} and \emph{updateRepo} (see \autoref{subfig:repo-rules} for a detailed illustration) to explain the static and dynamic analysis.
We explain the analysis for the overall API in \autoref{sec:graphql}.
\begin{figure}[t]
	\centering
	\includegraphics[scale=.5]{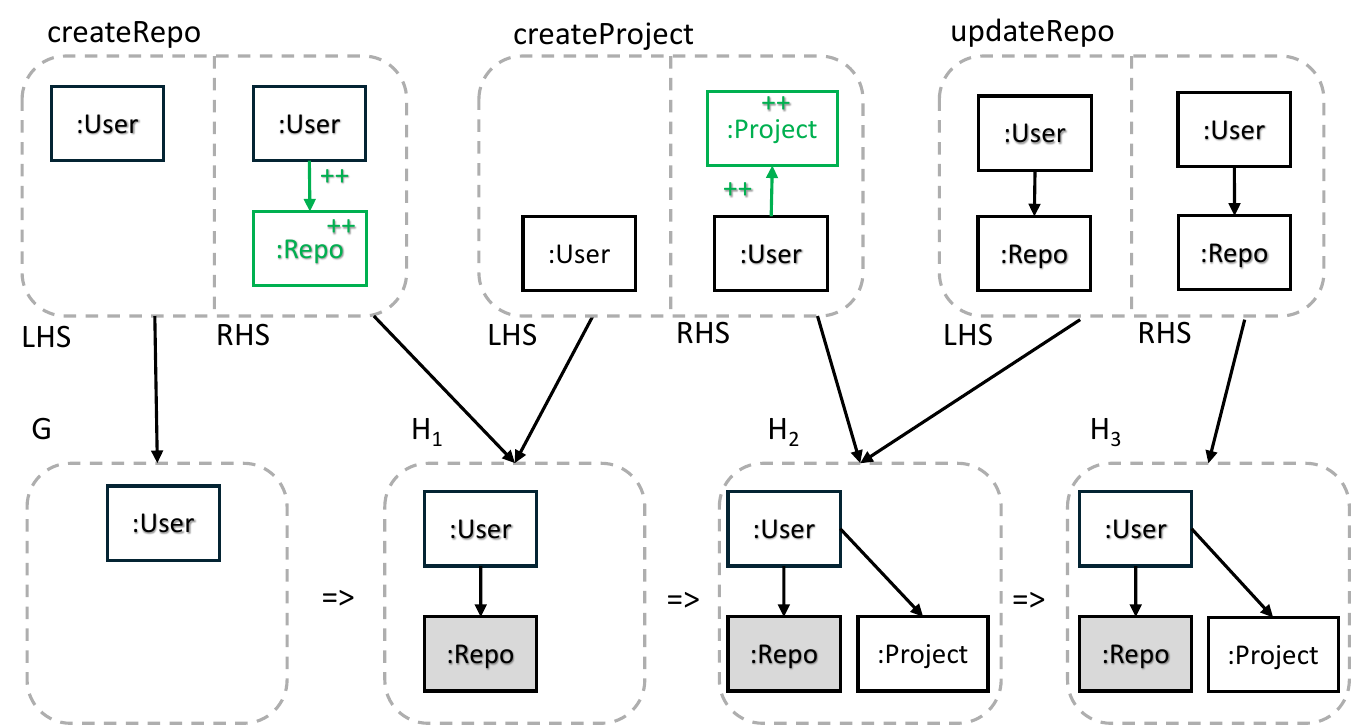}
	\caption{\label{fig:reducibleVulnerability} Graph API execution consisting of a sequence of three Graph API calls using the plain notation; the calls \emph{createRepo, createProject, updateRepo} are depicted on top, whereas the graphs $G, H_1, H_2, H_3$ at the bottom depict the results following the execution of each call}
\end{figure}

\subsubsection{Access control}
In a Graph API, access control\ie the fact that certain API calls that may create or manipulate sensitive data can only be executed by users with sufficient privileges, can be implemented via a role-based policy. In that case, users of the Graph API can only execute certain calls if they belong to a group with the required privileges. For instance, according to the GitHub access control policy, an API user with the role (or, in GitHub terms, \emph{permission level}) \emph{owner} has full control over the repository; an API user with the role \emph{collaborator} has only read and write access to the code and cannot, for example, delete the repository~\cite{GitHub__Permissionlevelspersonalaccountrepository}.

We formalize role-based access control policy and API calls whose access control is role-based below. The policy formalization is deliberately kept high-level, since we want to support cases where only rather informal policy descriptions are given. We merely assume that a security analyst is able to decide based on their expectations and on the policy description for each API call executed by some user with a specific role if said call is allowed or not. We call the policy formalization \emph{policy oracle}. For a call executed with a role, the policy oracle evaluates to true (false) if access (no access) should be granted. Below, we also present the formalization of the policy implementation. This formalization adheres to the same signature of the policy oracle and is similarly kept high-level, since we want to support APIs regardless of the specifics of their access control implementation. The policy implementation decides whether access is granted to the actual implementation of the call and is later used to formalize the notion of broken access control vulnerability.

\begin{figure}
	\begin{subfigure}[t]{\textwidth}
		\centering
		\includegraphics[width=.58\textwidth]{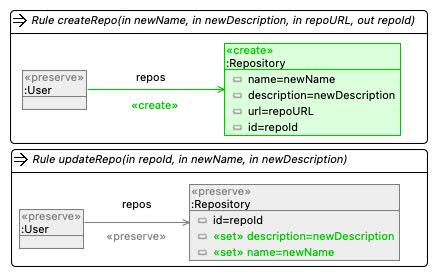}
		\caption{\label{subfig:repo-rules}Rules related to \texttt{Repository}.}
	\end{subfigure}
	\begin{subfigure}[t]{\textwidth}
		\centering
		\includegraphics[width=.9\textwidth]{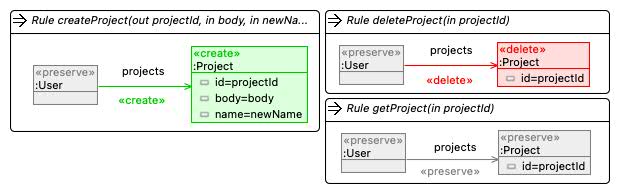}
		\caption{\label{subfig:project-rules}Rules related to \texttt{Project}.}
	\end{subfigure}
	\begin{subfigure}[t]{\textwidth}
		\centering
		\includegraphics[width=.5\textwidth]{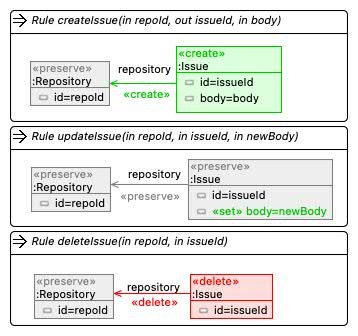}
		\caption{\label{subfig:issue-rules}Rules related to \texttt{Issue}.}
	\end{subfigure}
	\caption{\label{fig:rules} The sources (\emph{create} rules) and sinks (\emph{update}, \emph{get}, \emph{delete} rules) used in the running example and \autoref{sec:graphql}.}
\end{figure}

\begin{defi}[role-based Graph API call, policy oracle, policy implementation]\label{def:role-based-graphAPI}
	Given a Graph API $\mathcal{R}$ with schema $TG$.
	A \emph{role specification} $(RO, \leq_{RO})$ consists of a finite set of roles $RO$ and a partial order $\leq_{RO} \subseteq RO \times RO$.

	A \emph{role-based Graph API execution} $(t : G_0 \stackrel{r_1}{\Rightarrow} G_1 \stackrel{r_2}{\Rightarrow} G_2 \stackrel{r_3}{\Rightarrow} G_3 \ldots, role_t)$ is a Graph API execution such that each direct transformation in $t$ is executed with a specific role as given in
	$role_t: \mathbb{N} \rightarrow \mathcal{P}(RO)$, where $role_t(i)$ refers to the subset of roles of the i-th direct transformation $G_i \stackrel{r_{i+1}}{\Longrightarrow} G_{i+1}$.
	A \emph{role-based Graph API call} is a direct transformation within a role-based Graph API execution.
	We denote with $Sem^r(\mathcal{R})$ the set of all \emph{role-based Graph API executions}, and with $Sem^r_{call}(\mathcal{R})$ the set of all \emph{role-based Graph API calls}.

	A \emph{role-based access control policy oracle} $ACP = ((RO, \leq_{RO}), P)$ consists of a role specification $(RO, \leq_{RO})$ and a \emph{policy oracle} $P$\ie a total mapping $P: Sem^r_{call}(\mathcal{R}) \rightarrow \{true,false\}$.

	A \emph{role-based access control policy imlementation} $ACP^I = ((RO, \leq_{RO}), P^I)$ consists of a role specification $(RO, \leq_{RO})$ and a \emph{policy implementation} $P^I$\ie a total mapping $P^I: Sem^r_{call}(\mathcal{R}) \rightarrow \{true,false\}$ describing for a specific implementation $I$ of the Graph API calls if access is actually granted or not, respectively.
\end{defi}
Note that each call in an execution $t$ may be executed with a different role, which mirrors the typical interaction with a Graph API in practice---see \autoref{subsec:graphql_github}.

As discussed in \autoref{subsec:preliminaries_access-control}, a broken access control vulnerability occurs if policy requirements about user access to information and resources are not fulfilled. In our running example, such a vulnerability occurs if users without the proper privileges get access or are able to manipulate specific elements in the repositories\eg issues or projects, or repositories themselves; this would result to a data leak or even the loss of valuable information in the code repository.
Formally, a \emph{broken access control vulnerability} is captured by a role-based Graph API execution containing a call that breaks the access control policy.  This either happens if the call is executed with the specified role in a concrete policy implementation $P^I$, although this should not be possible according to the policy oracle $P$, or also if the call is not allowed with the specified role in a concrete policy implementation $P^I$, although this should be possible according to the policy oracle $P$.
\begin{defi}[Broken Access Control (BAC) vulnerability]\label{def:vulnerability}
	Given a graph API $\mathcal{R}$ with schema $TG$, a role-based access control policy oracle $ACP = ((RO, \leq_{RO}), P)$, and a policy implementation $P^I: Sem^r_{call}(\mathcal{R}) \rightarrow \{true,false\}$ associating the actual policy result to each API call.
	A \emph{Broken Access Control (BAC) vulnerability} is represented by a role-based Graph API execution $(t : G_0  \stackrel{r_1}{\Rightarrow} G_1 \stackrel{r_2}{\Rightarrow} G_2 \stackrel{r_3}{\Rightarrow} G_3 \ldots, roles_t)$ entailing an API call $c: G_i \stackrel{r_{i+1}}{\Longrightarrow} G_{i+1}$ such that $P(c)\neq P^I(c)$.
\end{defi}%

See \autoref{fig:reducibleVulnerability} again for an example of a Graph API execution.  Assume that the first and third call are executed by a \emph{user} as an owner of the repository and the second call by a \emph{collaborator}. For such a role-based Graph API execution, the security analyst may expect that access is granted for each of the calls. So the policy oracle would map to true for each of the calls in the role-based Graph API execution.  If the actual implementation of the Graph API execution raises a BAC exception and does not allow one of the calls (while being executed with the specified roles), we have a BAC vulnerability. In this case, the policy implementation is too conservative w.r.t. the expectations of the security analyst for this role-based Graph API execution to access control.  Imagine now that a repository can only be updated by a user of the repository with the role \emph{owner} or \emph{collaborator}, but not by any other arbitrary user.  Then, if the actual implementation of the Graph API execution does not raise a BAC exception when executing this third call with a role other than \emph{owner} or \emph{collaborator}, this again represents a BAC vulnerability. In this case, the access control policy implementation is too liberal w.r.t. the expectations of the security analyst for this role-based Graph API execution, and we have discovered a data leak.

We say that an access control policy oracle is \emph{stable under shift} if for each pair of sequentially independent calls, the policy oracle mapping for each of the calls remains identical when swapping the calls. Recall that a pair of graph transformation steps $ G \stackrel{r}{\Longrightarrow} G' \stackrel{r'}{\Longrightarrow} G''$ is \emph{sequentially independent} if their order can be reversed\ie
$G \stackrel{r'}{\Longrightarrow} H \stackrel{r}{\Longrightarrow} G''$ without affecting the end result $G''$.%
\begin{defi}[policy stable under shift]\label{def:stable-under-shift}
	A role-based access control policy oracle $ACP = ((RO, \leq_{RO}), P)$ is \emph{stable under shift} if for each pair of sequentially independent calls $G \stackrel{r}{\Longrightarrow} G' \stackrel{r'}{\Longrightarrow} G''$ executed with roles $ro,ro' \in \mathcal{P}(RO)$ respectively, it holds that $P(G \stackrel{r}{\Longrightarrow} G') = P(H \stackrel{r}{\Longrightarrow} G'')$ and $P(G' \stackrel{r'}{\Longrightarrow} G'') = P(G \stackrel{r'}{\Longrightarrow} H)$, the latter executed with roles $ro,ro' \in \mathcal{P}(RO)$ respectively. Analogously, a policy implementation $P^I$ is stable under shift, if this property also holds for the policy implementation.
\end{defi}

As policy oracles and implementations discussed in this article are assumed to disallow administrative actions (see \autoref{subsec:preliminaries_access-control}), they are stable under shift\ie swapping the order of two calls does not change the policy oracle mapping for said calls. \autoref{def:stable-under-shift} frames this characteristic in terms of graph transformation and is further used in a relevant proof later in the paper---see \autoref{thm:general}. A policy that would allow administrative actions\ie a call could assign the \emph{collaborator} role to a user, could possibly not satisfy this notion of stability\eg whether such an administrative call precedes or succeeds another call could change the policy mapping for said calls.
\subsubsection{Tainting}
The idea of the taint analysis is to systematically trace the flow of sensitive data through the Graph API with the aim of detecting potential BAC vulnerabilities.
The aim of the static taint analysis is to detect potential Graph API interactions that may not be properly secured.
The aim of the dynamic taint analysis is to systematically further test these risky interactions to detect if BAC vulnerabilities arise indeed.

To this extent, the security analysis first needs to taint those nodes in the schema representing sensitive data objects.
A BAC vulnerability is related to such a tainted node if a Graph API call introduces an instance node of a tainted node, and a later call reads or manipulates this node without the necessary privileges; a direct BAC vulnerability occurs when the introducing call is directly followed by the reading or manipulating call.
Based on the schema knowledge about tainted nodes, we also derive a tainted Graph API.
The tainted API distinguishes source rules (creating sensitive data) from sink rules (reading or manipulating sensitive data).
Tainted nodes and tainted graph APIs form the basis for the automatic static analysis.
It will detect potential interactions between pairs of source and sink rules representing risky flows w.r.t. passing sensitive data.
In this way, the static analysis will allow us to find all potential direct BAC vulnerabilities related to tainted nodes.
Moreover, we will argue under which conditions we can also find specific indirect BAC vulnerabilities.
We formalize these tainting concepts for our approach below.

\begin{defi}[tainted type graph, tainted nodes and their relation to a (direct) BAC vulnerability]\label{def:taint}
	Let $TG=(V_{TG},E_{TG},s_{TG},t_{TG})$ be a type graph. Then, $TG_t=(TG,T)$ is a \emph{tainted type graph} with $T \subseteq V_{TG}$ the set of \emph{tainted nodes}.

	Given a Graph API $\mathcal{R}$ with schema $TG$, a \emph{role-based access control policy oracle} $ACP = ((RO, \leq_{RO}), P)$ and policy implementation $ACP^I = ((RO, \leq_{RO}), P^I)$.
	The role-based Graph API execution $(t : G_0  \stackrel{r_1}{\Rightarrow} G_1 \stackrel{r_2}{\Rightarrow} G_2 \stackrel{r_3}{\Rightarrow} G_3 \ldots, roles_t)$ represents a {BAC vulnerability related to a tainted node} $s$ from $T$ if it entails a call $c: G_i \stackrel{(m_{i+1},r_{i+1})}{\Longrightarrow} G_{i+1}$ such that a node $n$ with $type_{V}(n) = s$ in the image of $m_{i+1}$ exists, leading to $P(c) \neq P^I(c)$, and $n$ does not belong already to $G_0$.
	A BAC vulnerability $(t, roles_t)$ is called \emph{direct} if the node $n$ with $type_{V}(n) = s$ is created in the call $G_{i-1} \stackrel{(m_{i},r_{i})}{\Longrightarrow} G_{i}$ in $t$ directly before $c: G_i \stackrel{(m_{i+1},r_{i+1})}{\Longrightarrow} G_{i+1}$.
\end{defi}

See \autoref{fig:typeGraph} for an example of a tainted schema, where we taint the node \texttt{Repository}---tainting is illustrated by the gray background color of the node in question.
In this case, this node is tainted because we want to further analyze if users without the proper privileges are able to manipulate created repositories. Assume that a BAC vulnerability arises in the implementation of the exemplary role-based Graph API execution in~\autoref{fig:reducibleVulnerability}, since it is possible to execute the third call successfully as a user with a role that does not have the required privileges.  Then we say that this is a BAC vulnerability related to the tainted repository node $R$.  This vulnerability is \emph{not a direct vulnerability}, as the \texttt{Repository} node has been created in the first call and not in the previous\ie second, call.

Finally, to obtain a tainted graph API, we define \emph{source rules}\ie rules that create instances of a tainted node in a schema, and \emph{sink rules}\ie  rules that read or manipulate instances of tainted nodes.

\begin{defi}[source rule, sink rule, tainted graph API]\label{def:tainted-graph-API}
	Given a tainted type graph $TG_t=(TG,T)$ and a Graph API $\mathcal{R}$ with schema $TG$.
	A rule $r = (L  \stackrel{le}{\hookleftarrow} K \stackrel{ri}{\hookrightarrow} R)$ from $\mathcal{R}$ is a \emph{source rule} if there exists a node $n$ in $R \setminus L$ such that $type_{R,V}(n) \in T$.
	A rule $r = (L  \stackrel{le}{\hookleftarrow} K \stackrel{ri}{\hookrightarrow} R)$ from $\mathcal{R}$ is a \emph{sink rule} if there exists a node $n$ in $L$ such that $type_{L,V}(n) \in T$.
	A \emph{tainted Graph API} $(TG_t,\mathcal{R},src,sink)$ consists of a tainted type graph $TG_t$, a Graph API typed over $TG$, and a corresponding set of source and sink rules, $src$ and $sink$.
\end{defi}

As a minimal example of a tainted Graph API comprising a set of source and sink rules, see the rules \textit{createRepo} and \textit{updateRepo} in \autoref{subfig:repo-rules}.
The rule \textit{createRepo} creates a tainted node\ie an instance of \texttt{Repository}, and is therefore a source rule, whereas \textit{updateRepo} refers to instances of \texttt{Repository} and is therefore a sink rule.  As we will explain in the next section, our static analysis identifies such pairs of source and sink rules such that the security analyst can review the policy description to check if the access to sensitive (i.e. tainted) nodes is properly secured.

\begin{figure}[t]
	\centering
	\includegraphics[width=.45\textwidth]{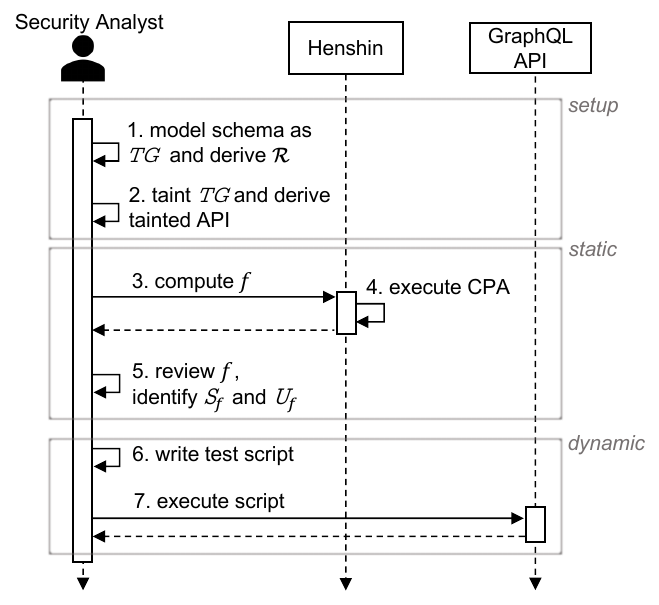}
	\caption{\label{fig:overview} An overview of the sequence of steps in the presented approach.}
\end{figure}
\subsection{Static Analysis}\label{subsec:analysis_static}
\autoref{fig:overview} shows an overview of the sequence of steps in the presented approach to taint analysis. The setup segment comprised steps 1 and 2. The next segment is the static analysis, which is based on the output of the setup.
The static analysis comprises the following steps: the provision of the schema and tainted Graph API obtained during setup by the security analyst (step 3 in \autoref{fig:overview}); the execution of the actual static analysis over the input provided by the security analyst, which is performed based on available tool support (step 4); and the review of output of the analysis by the security analyst (step 5) so that potential BAC vulnerabilities are identified due to an incomplete or incorrect access control policy. We describe this segment in detail below.

Owing to the graph-transformation-based formalization, the static analysis can automatically determine potential dependencies between the rules for a tainted Graph API by using CPA---see \autoref{subsec:preliminaries_graph-transformation}. In particular, we use the coarse granularity level of the  \emph{multi-granular static dependency analysis} to determine all minimal \emph{dependency reasons}.
The \emph{tainted information flow} for a tainted Graph API is then a symbolic yet complete description of all possibilities to use a tainted node directly after it has been created. In the running example, a dependency occurs when a \texttt{Repository} is created, which is afterward updated. The unique dependency reason between rules \textit{createRepo} and \textit{updateRepo} in \autoref{subfig:repo-rules} consists of the overlap of the complete right-hand side of the rule \textit{createRepo} with the left-hand side of the rule \textit{updateRepo}.
\begin{defi}[tainted information flow]\label{def:tainted-information-flow}
	Let $(TG_t,\mathcal{R},src,sink)$ be a tainted Graph API with $TG_t = (TG, T)$. We denote with  $src(s)$ ($sink(s)$) the subset of rules in $src$ ($sink$) creating (using) a node of type $s$ from $T$, respectively. Then, the \emph{tainted information flow} of the API is given by $f = \cup \{DR(r_{src},r_{sink})| r_{src} \in src(s), r_{sink} \in sink(s), s\in T\}$ with
	$DR(r_{src},r_{sink})$ the set of dependency reasons for the rule pair $(r_{src}, r_{sink})$.
\end{defi}

\autoref{subfig:henshin-output} shows the tainted information flow for the tainted \gapi from the running example---see dependency between \emph{createRepo} and \emph{updateRepo}.

The following theorem states that each BAC vulnerability related to a tainted node (see \autoref{def:taint}) is represented by a sequence of API calls containing a source rule application creating the tainted node followed at some point by a sink rule which uses the tainted node in a manner that violates the policy. We also say that such a BAC vulnerability entails a tainted flow instance.

\begin{thm}[BAC vulnerability and tainted flow instance]\label{thm:vulnerability-taint}
	Let $(t : G_0  \stackrel{r_1}{\Rightarrow} G_1 \stackrel{r_2}{\Rightarrow} G_2 \stackrel{r_3}{\Rightarrow} G_3 \ldots, roles_t)$ be a BAC vulnerability \emph{related to a tainted node $s$}. Then, the BAC vulnerability entails a source rule application via some rule $r_{src}$ in $src(s)$ creating a node $n$ with $type_{V}(n) = s$ and a sink rule application via a sink rule $r_{sink}$ in $sink(s)$ using $n$. We call this subsequence of $t$ a \emph{tainted flow instance}.
\end{thm}
\begin{proof}
	There exists an API call $c: G_i \stackrel{(m_{i+1},r_{i+1})}{\Longrightarrow} G_{i+1}$ in $t$ such that $P(c) \neq P^I(c)$ for a node $n$ in the image of $m_{i+1}$ with $type_{V}(n) = s$ (\autoref{def:taint}).
	We thus have that $r_{i+1}$ equals a sink rule $r_{sink}$ in $sink(s)$, since the preimage of $n$ in the left-hand side of the rule $r_{i+1}$ then also has type $s$ (\autoref{def:tainted-graph-API}).
	Since the node $n$ is not already contained in $G_0$ (\autoref{def:taint}), there must be a call $G_{j} \stackrel{(m_{j+1},r_{j+1})}{\Longrightarrow} G_{j+1}$ with $0 \leq j < i$ creating $n$ in $t$.
	This means that $r_{j+1}$ equals a source rule $r_{src}$ in $src(s)$, since it creates node $n$ of tainted type $s$ (\autoref{def:tainted-graph-API}).
	Thus we can rewrite $t$ into $G_0 \ldots \stackrel{(m_{j+1},r_{j+1})}{\Longrightarrow} G_{j+1} \stackrel{r_{j+2}}{\Rightarrow} \ldots G_i \stackrel{(m_{i+1},r_{i+1})}{\Longrightarrow} G_{i+1} \ldots$ with $r_{j+1} = r_{src}$ and $r_{i+1} = r_{sink}$.
\end{proof}

Based on the definition of a tainted flow instance, we can now technically define a potential BAC vulnerability: such a vulnerability is an API execution that contains a tainted flow instance, although it has not yet been verified whether the vulnerability is covered by the policy oracle implementation\ie whether for each call $c$ in the API execution $P(c) = P^I(c)$.

Based on \autoref{thm:vulnerability-taint}, we can distinguish BAC vulnerabilities with atomic and non-atomic\ie more complex, tainted flow instances. The former occurs when the sink rule application violating the policy oracle is caused by a source rule application without the occurrence of any intermediate applications of some sink or source rule.
\begin{defi}[atomic tainted flow instance]\label{def:atomic-flow}
	Given a BAC vulnerability $(t : G_0  \stackrel{r_1}{\Rightarrow} G_1 \stackrel{r_2}{\Rightarrow} G_2 \stackrel{r_3}{\Rightarrow} G_3 \ldots, roles_t)$ \emph{related to a tainted node $s$} such that $t : G_0 \ldots \stackrel{(m_{j+1},r_{j+1})}{\Longrightarrow} G_{j+1} \stackrel{r_{j+2}}{\Rightarrow} \ldots G_i \stackrel{(m_{i+1},r_{i+1})}{\Longrightarrow} G_{i+1} \ldots)$ with $G_i \stackrel{(m_{i+1},r_{i+1})}{\Longrightarrow} G_{i+1}$ violating the policy oracle and $r_{j+1} = r_{src}$ and $r_{i+1} = r_{sink}$. If all intermediate rule applications between $G_j \stackrel{(m_{j+1},r_{j+1})}{\Longrightarrow} G_{j+1}$ and $G_i \stackrel{(m_{i+1},r_{i+1})}{\Longrightarrow} G_{i+1}$ in $t$ are of rules other than $r_{src}$ and $r_{sink}$, then we have a BAC vulnerability containing an \emph{atomic tainted flow instance}.
\end{defi}

A security vulnerability detection technique is \emph{sound} for a category of vulnerabilities if it can correctly conclude that a program has no vulnerabilities of that category. In the following theorem, we argue that a static analysis technique computing the tainted information flow is sound\ie it detects all \emph{direct} potential BAC vulnerabilities related to tainted nodes.%
\begin{thm}[soundness for direct vulnerabilities]\label{thm:direct}
	Given a tainted type graph $TG_t = (TG,T)$, a tainted Graph API $(TG_t,\mathcal{R},src,sink)$, and a tainted information flow  $f = \cup \{DR(r_{src},r_{sink})| r_{src} \in src(s),  r_{sink} \in sink(s), s\in T\}$.
	Given a \emph{direct BAC vulnerability} $(t : G_0 \stackrel{r_1}{\Rightarrow} G_1 \stackrel{r_2}{\Rightarrow} G_2 \stackrel{r_3}{\Rightarrow} G_3 \ldots, roles_t)$ related to a tainted node $s$ from $T$, then
	$DR(r_{src},r_{sink})$ in $f$ is non-empty for some $r_{src}$ in $src(s)$ and $r_{sink}$ in $sink(s)$.
\end{thm}
\begin{proof}
	Based on \autoref{thm:vulnerability-taint} and given that $t$ is a direct vulnerability (\autoref{def:taint}) we can rewrite $t$ into $G_0 \ldots \stackrel{(m_{i},r_{i})}{\Longrightarrow} G_{i} \stackrel{(m_{i+1},r_{i+1})}{\Longrightarrow} G_{i+1} \ldots$ with $r_{i} = r_{src} \in src(s)$ and $r_{i+1} = r_{sink} \in sink(s)$.
	As dependency reasons are complete~\cite{Lambers_2019_Granularityconflictsdependenciesgraphtransformationsystemstwodimensionalapproach}, there exists a dependency reason $d$ in $DR(r_{src},r_{sink})$ such that a node $m$ of type $s$ (instantiated to node $n$ in $t$) is created by $r_{src}$ and used by $r_{sink}$ as described in $d$.
\end{proof}

Our detection technique is not sound for \emph{indirect vulnerabilities}\ie situations where a tainted information flow occurs because of applications of intermediate rules between a source rule and a sink rule. Unsoundness may occur for two reasons---of which we give examples in \autoref{exa:indirect-delete} and \autoref{exa:indirect-create}. First, the source rule may create a tainted node together with some incident graph structure; this incident graph structure needs to be deleted by an intermediate rule in order for the sink rule that deletes the tainted node to become applicable. The second reason is that an intermediate rule may create an incident graph structure to the tainted node thereby enabling the application of a sink rule requiring this incident graph structure.
\begin{exa}[potential indirect vulnerability due to incident structure deleted]\label{exa:indirect-delete}
	The toy example in \autoref{fig:deleteIncident} illustrates a potential indirect vulnerability, since the rule \emph{deleteT} deletes the tainted node $T$ created by the rule \emph{createIncidentT}.
	\begin{figure}
		\begin{subfigure}[t]{\textwidth}
			\centering
			\includegraphics[width=.55\textwidth]{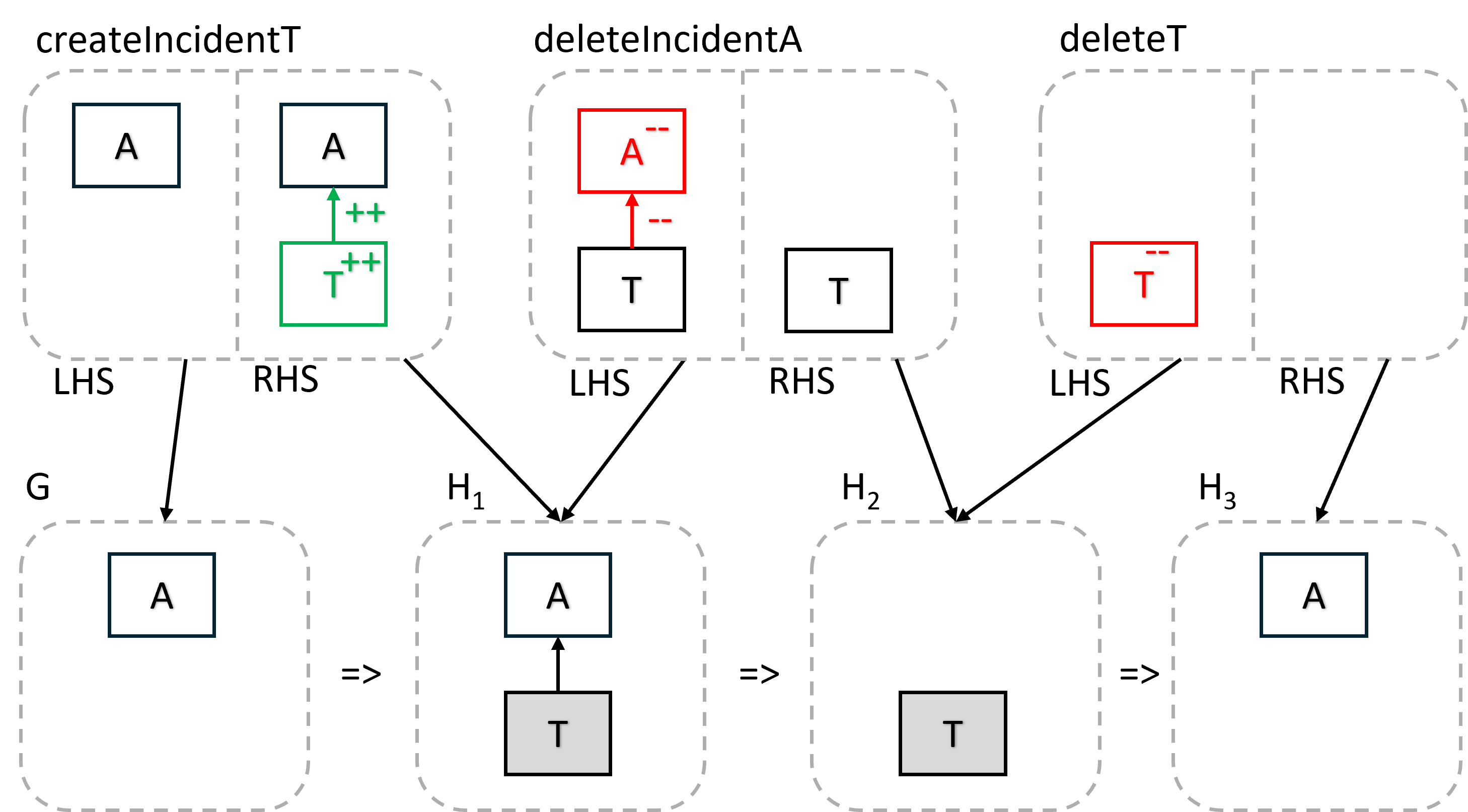}
			\caption{Potential indirect vulnerability with forbidden incident structure deleted first.}
			\label{fig:deleteIncident}
		\end{subfigure}
		\begin{subfigure}[t]{\textwidth}
			\centering
			\includegraphics[width=.45\textwidth]{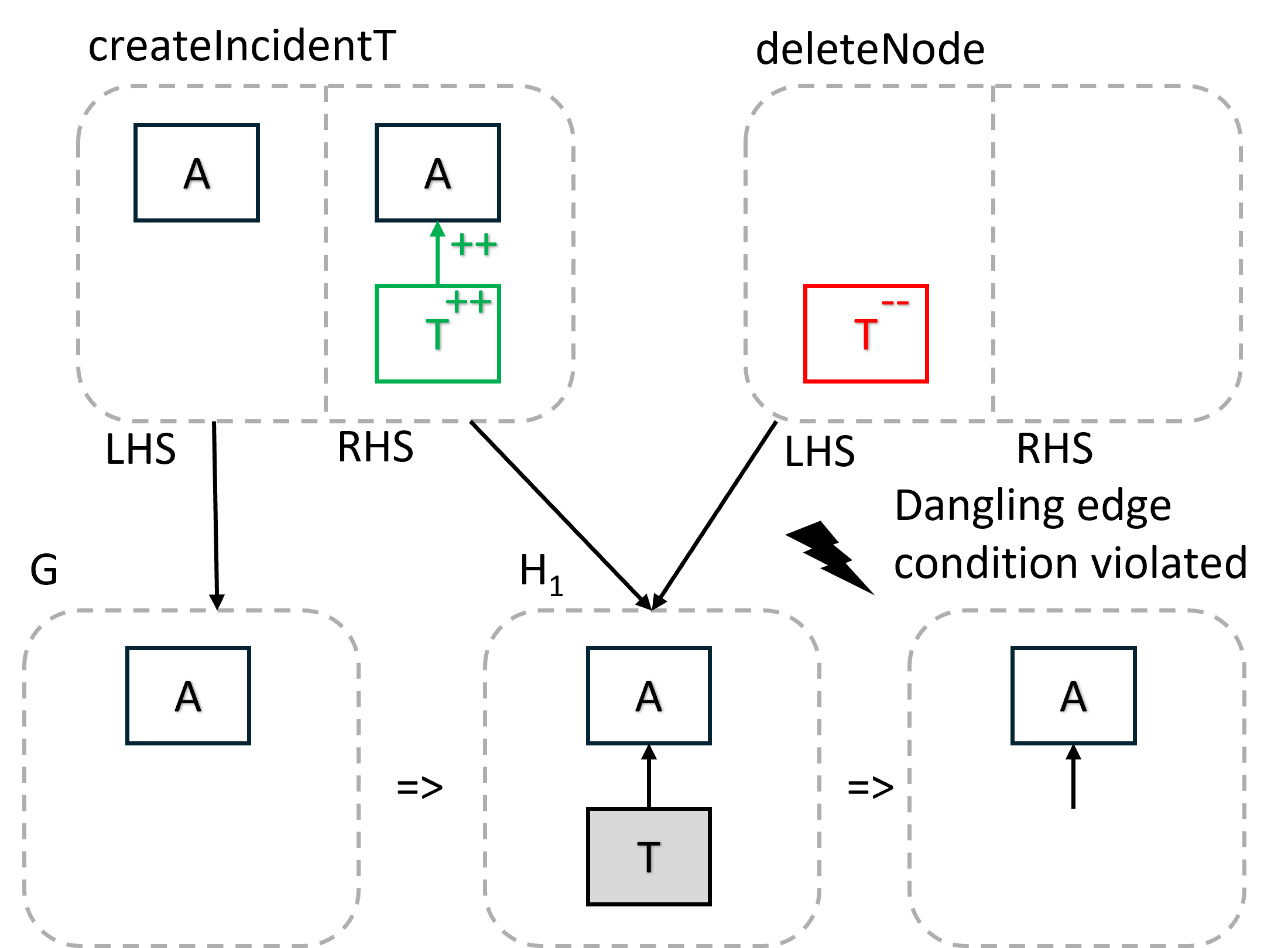}
			\caption{No potential direct vulnerability because of dangling edge to node $A$.}
			\label{fig:deleteIncidentDangling}
		\end{subfigure}
		\caption{\label{fig:deleteIncidentWhole}Illustrations for \autoref{exa:indirect-delete}.}
	\end{figure}
	The rule \emph{deleteT} is enabled by the application of the intermediate rule \emph{deleteIncidentA}. As the sink rule \emph{deleteT} is not directly applicable after the source \emph{createIncidentT} (as this would create a dangling edge incident to node $A$---see \autoref{fig:deleteIncidentDangling}), there is no critical pair\ie dependency reason, between the two rules.
\end{exa}

\begin{exa}[Potential indirect vulnerability due to incident structure created]\label{exa:indirect-create}
	The toy example in \autoref{fig:createIncident} illustrates a potential indirect vulnerability, since the rule \emph{createIncidentB+} uses the tainted node $T$ created by rule \emph{createIncidentT}.
	The application of \emph{createIncidentB+} is enabled by the intermediate application of \emph{createIncidentB}. As the sink rule \emph{createIncidentB+} is not directly applicable after the application of the source rule \emph{createIncidentT} (as this would create a dangling edge incident to node $A$---see \autoref{fig:createIncidentDangling}), there is no critical pair\ie dependency reason, between the two rules.

	\begin{figure}
		\begin{subfigure}[t]{\textwidth}
			\centering
			\includegraphics[width=.55\textwidth]{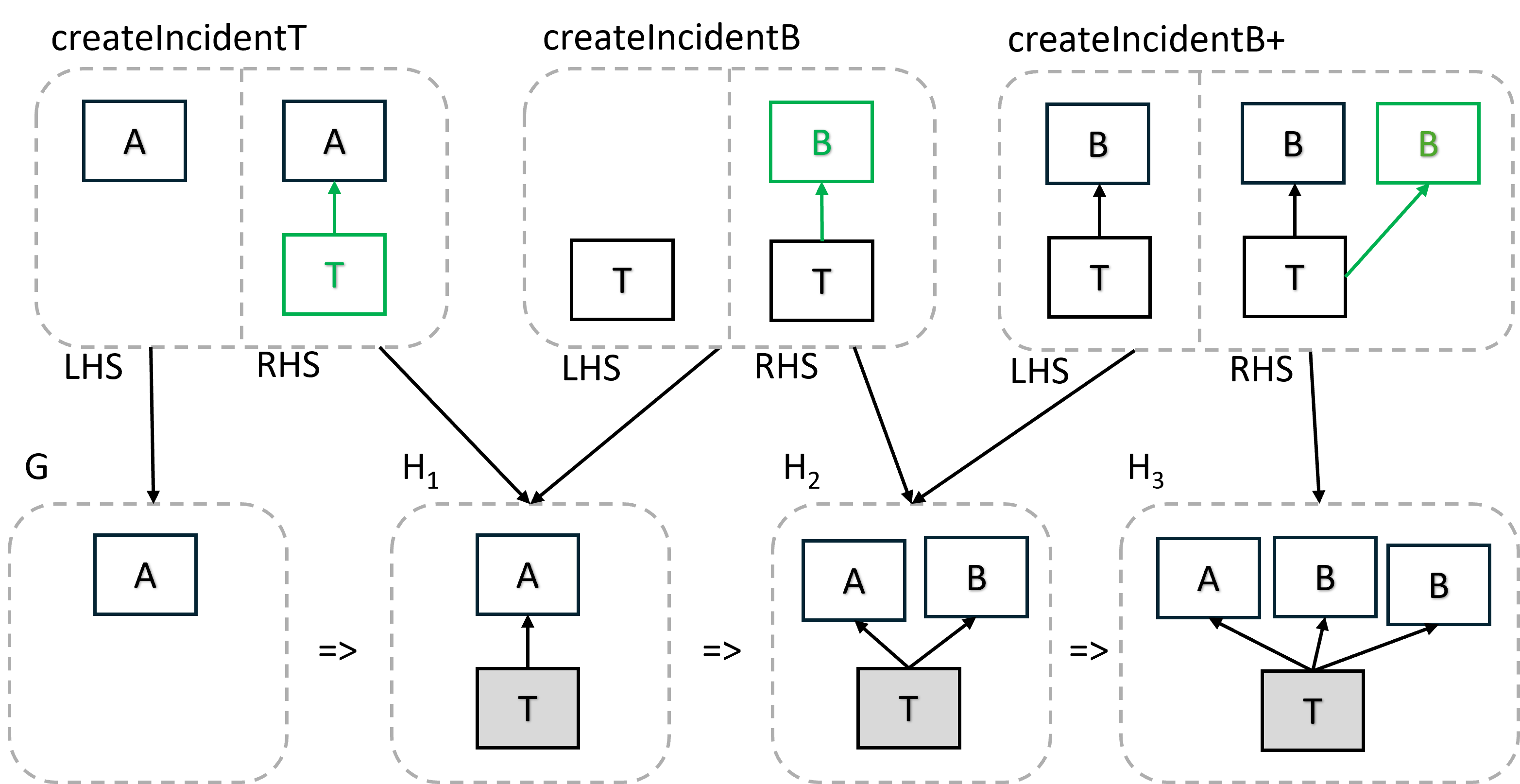}
			\caption{Potential indirect vulnerability with required incident structure created first.}
			\label{fig:createIncident}
		\end{subfigure}
		\begin{subfigure}[t]{\textwidth}
			\centering
			\includegraphics[width=.5\textwidth]{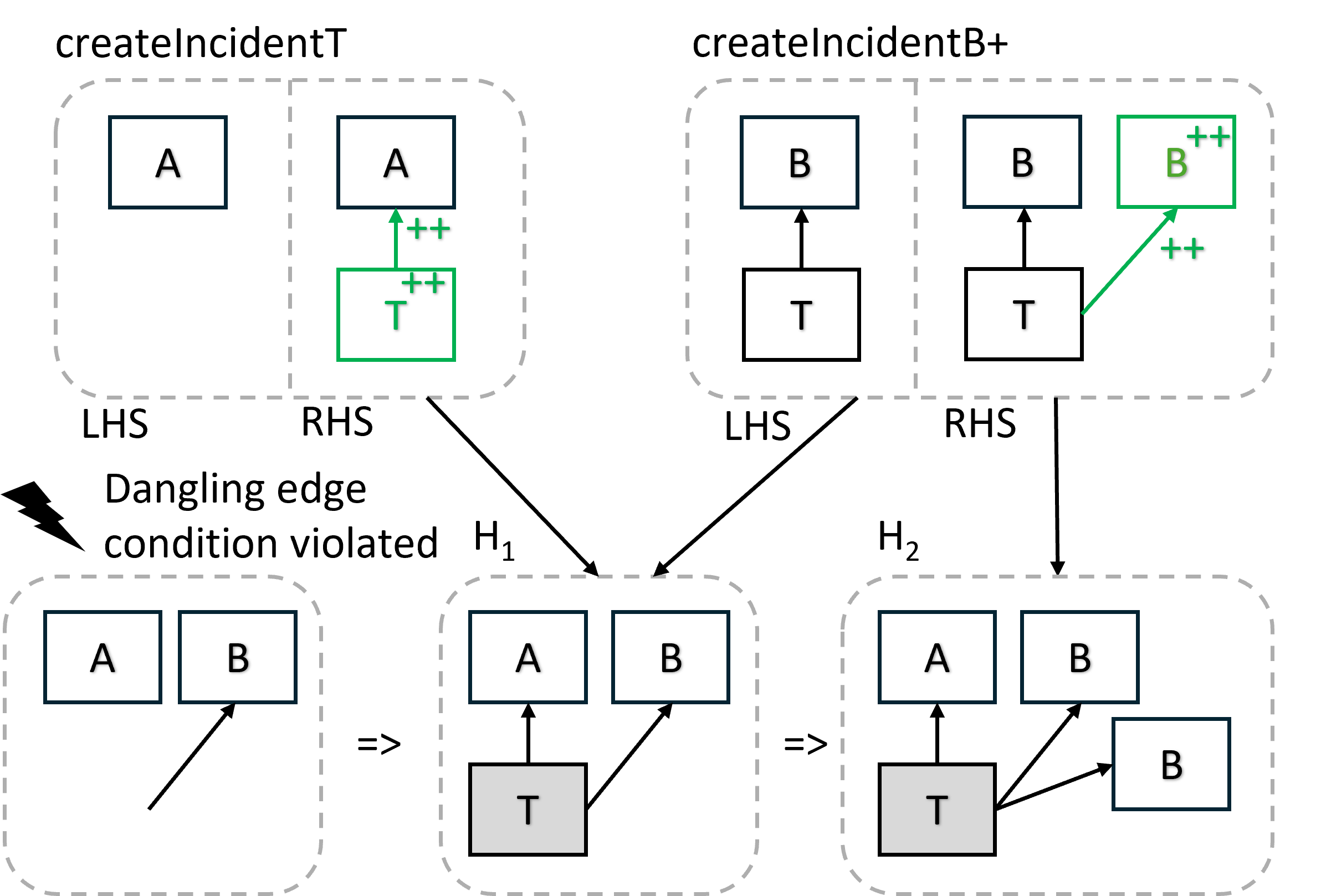}
			\caption{No potential direct vulnerability because of dangling edge to node $B$.}
			\label{fig:createIncidentDangling}
		\end{subfigure}
		\caption{\label{fig:createIncidentWhole}Illustrations for \autoref{exa:indirect-create}.}
	\end{figure}
\end{exa}
The examples above illustrate that critical pairs (or dependency reasons) are effective at illustrating direct dependencies only. As soon as dependencies arise because of the application of more than one rule these compact dependency representations are not complete anymore.

The following theorem proves that under certain conditions our analysis is still sound w.r.t. finding \emph{both direct and indirect vulnerabilities with an atomic flow instance}. This is achieved by \emph{reducing} indirect vulnerabilities to direct ones, provided that these two conditions are met: we can shift all intermediate rule applications between a source and sink rule application to some application before the source or after the sink rule application; and, moreover, the access control policy oracle is stable under shift---see \autoref{def:stable-under-shift}.
A static check of whether shifting of rules can be achieved by checking the sequential independence of pairs of rules---see \autoref{def:stable-under-shift}; if a rule pair is sequentially independent, then the application of both rules would have the same result regardless of the order; thus, the transformations can be switched. Sequential independence can be computed by the binary level of the \textit{multi-granular static dependency analysis}~\cite{Born_2017_GranularityConflictsDependenciesGraphTransformationSystems,Lambers_2018_InitialConflictsDependenciesCriticalPairsRevisited}.

\begin{thm}[soundness for vulnerabilities]\label{thm:general}
	Given a tainted Graph API $(TG_t,\mathcal{R},src,sink)$ with $TG_t = (TG,T)$ and tainted information flow  $f = \cup \{DR(r_{src},r_{sink})| r_{src} \in src(s),  r_{sink} \in sink(s), s\in T\}$
	such that the following holds for each $(r_{src},r_{sink})$:
	(1) for each rule $r$ from $\mathcal{R} \setminus \{r_{sink}\}$ it holds that the rule pair $(r_{src},r)$ is sequentially independent
	or
	(2) for each rule $r$ from $\mathcal{R}\setminus \{r_{src}\}$ it holds that the rule pair $(r, r_{sink})$ is sequentially independent.

	Given a policy oracle $ACP = (RO, \leq_{RO}, P)$ and policy implementation $P^I$ which are stable under shift and a \emph{BAC vulnerability with atomic flow instance} represented by $(t : G_0 \stackrel{r_1}{\Rightarrow} G_1 \stackrel{r_2}{\Rightarrow} G_2 \stackrel{r_3}{\Rightarrow} G_3 \ldots, roles_t)$ related to a tainted node $s$ from $T$, then
	$DR(r_{src},r_{sink})$ in $f$ is non-empty for some $r_{src}$ in $src(s)$ and $r_{sink}$ in $sink(s)$.
\end{thm}
\begin{proof}
	Assume the BAC vulnerability $(t : G_0 \stackrel{r_1}{\Rightarrow} G_1 \stackrel{r_2}{\Rightarrow} G_2 \stackrel{r_3}{\Rightarrow} G_3 \ldots, roles_t)$ related to a tainted node $s$ from $T$.
	If the BAC vulnerability is direct, then, based on \autoref{thm:direct}, it follows that $DR(r_{src},r_{sink})$ in $f$ is non-empty for some $r_{src}$ in $src(s)$ and $r_{sink}$ in $sink(s)$.

	Let the BAC vulnerability be indirect. Then, according to \autoref{thm:vulnerability-taint}, we can rewrite $t$ into $G_0 \ldots \stackrel{(m_{j+1},r_{j+1})}{\Longrightarrow} G_{j+1} \stackrel{r_{j+2}}{\Rightarrow} \ldots G_i \stackrel{(m_{i+1},r_{i+1})}{\Longrightarrow} G_{i+1} \ldots$ with $r_{j+1} = r_{src} \in src(s)$ and $r_{i+1} = r_{sink} \in sink(s)$.
	Now assume that condition (1) holds, then we know that for each rule $r$ from $\mathcal{R} \setminus \{r_{j+1}\}$ it holds that $(r_{j+1},r)$ is sequentially independent. Because we have a vulnerability with an atomic flow instance, we know that all intermediate rule applications are applications of some rule $r$ from $\mathcal{R} \setminus \{r_{j+1},r_{i+1}\}$.
	We can thus obtain from $t$ a transformation
	$t': G_0 \ldots G \stackrel{r_{j+1}}{\Longrightarrow} G' \stackrel{r_{i+1}}{\Longrightarrow} G'' \ldots)$ with $roles_{t'})$.
	This is because we can stepwise switch toward the front each intermediate rule application with the application of the source rule $r_{j+1}$. Since the policy is stable under shift, from \autoref{thm:direct} it follows now for $(t', roles_{t'})$ that $DR(r_{src},r_{sink})$ in $f$ is non-empty for some $r_{src}$ in $src(s)$ and $r_{sink}$ in $sink(s)$.

	If condition (2) holds, we argue similarly, but instead switch each intermediate rule application of a rule other than the source rule toward the back.
\end{proof}

As illustrated by \autoref{exa:indirect-delete} and \autoref{exa:indirect-create}, the conditions from \autoref{thm:general} are not met in all cases. Moreover, indirect vulnerabilities can also occur because of the intermediate rules being themselves source and sink rules. The development of a new, sound static detection technique for all indirect dependencies constitute part of our future plans; such a technique requires an extension of the theory of critical pairs.

The following example shows that the conditions from \autoref{thm:general} applies to a part of our running example, illustrating how indirect vulnerabilities can be reduced to direct ones.
\begin{exa}[potential indirect vulnerability reducible to potential direct vulnerability]\label{exa:reducible-vulnerability}
\autoref{fig:reducibleVulnerability} illustrates a potential indirect vulnerability: a user creates a repository that is updated after the user creates a project.
This potential indirect vulnerability can be reduced to a direct one, since it is possible to switch the intermediate application of \emph{createProject} with the application of \emph{updateRepo}---see \autoref{fig:reducedVulnerability}.

\autoref{fig:irreducibleVulnerability} illustrates another potential indirect vulnerability: a user creates a repository for which an issue is created and afterward this issue is deleted again.  We cannot switch the intermediate rule application \emph{createIssue} forward or backward in this sequence. This example illustrates the need for a new, sound static detection technique for all indirect dependencies.
\end{exa}

\begin{figure}
	\begin{subfigure}[t]{\textwidth}
		\centering
		\includegraphics[width=.55\textwidth]{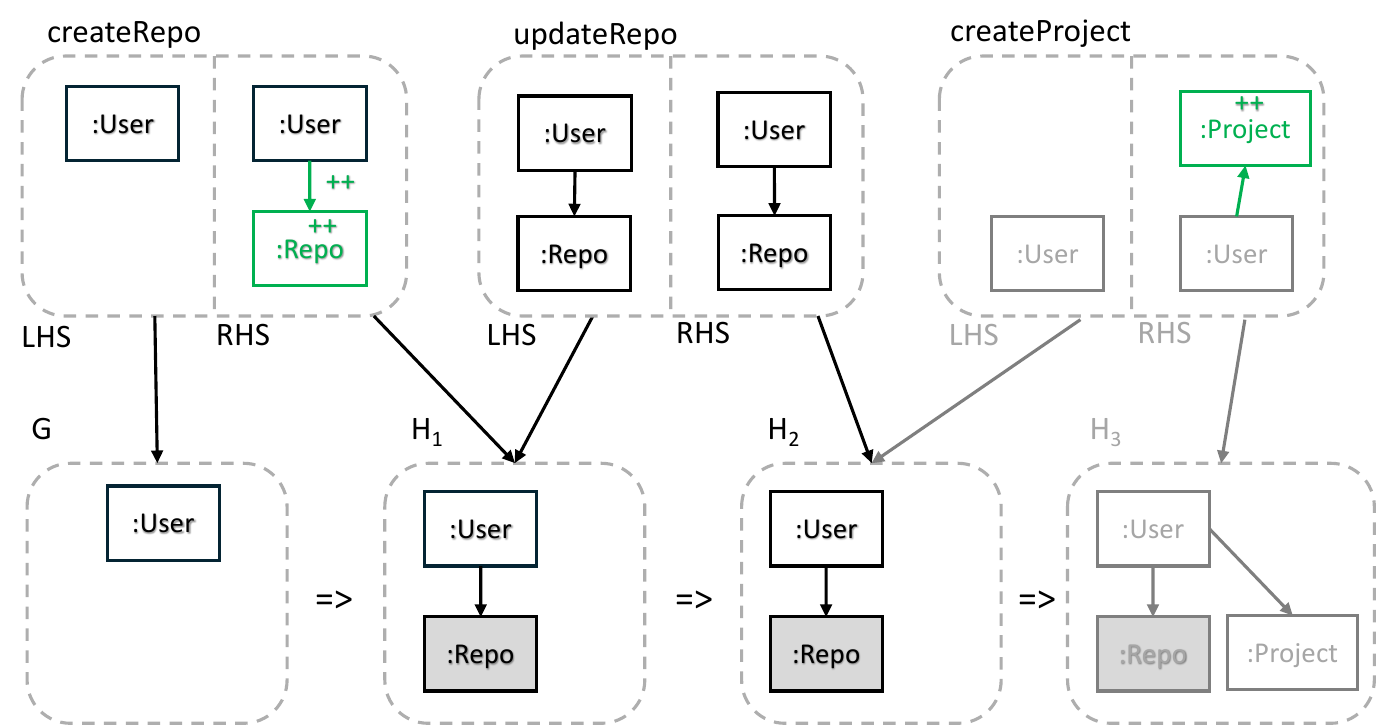}
		\caption{Potential indirect vulnerability reduced to potential direct vulnerability.}
		\label{fig:reducedVulnerability}
	\end{subfigure}
	\begin{subfigure}[t]{\textwidth}
		\centering
		\includegraphics[width=.55\textwidth]{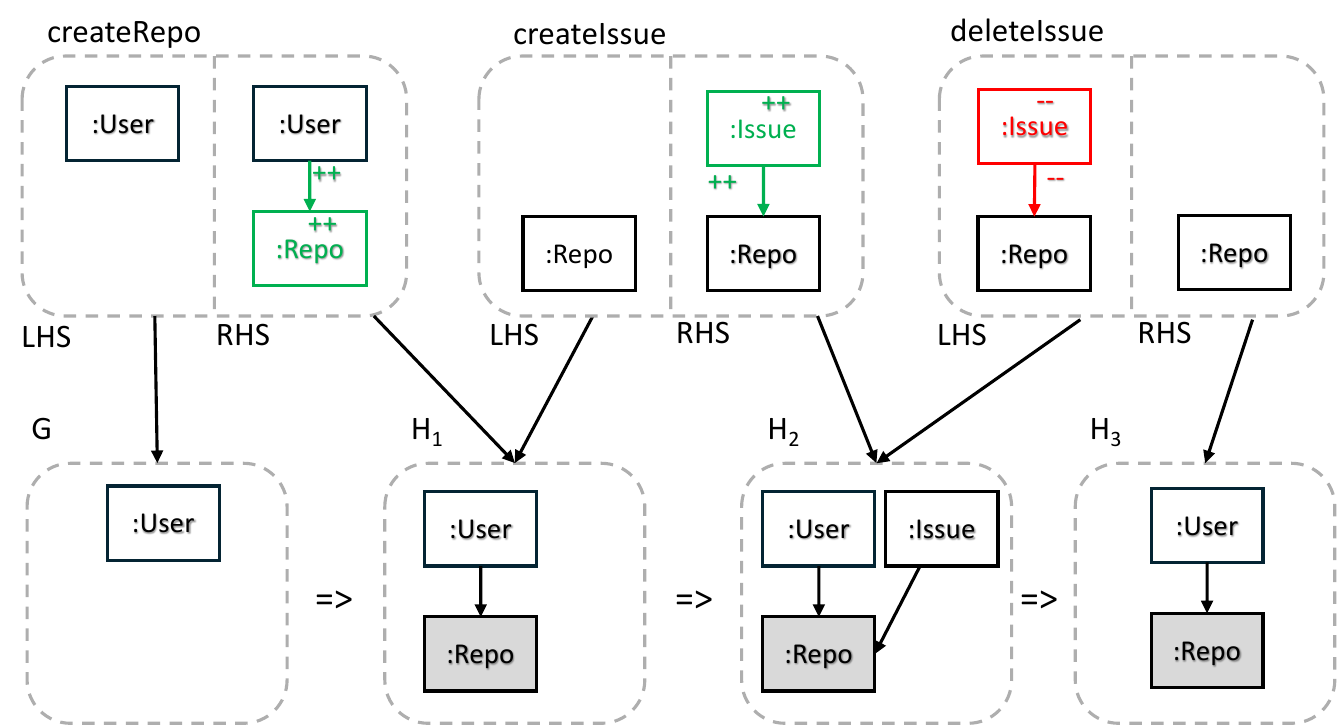}
		\caption{Potential indirect vulnerability.}
		\label{fig:irreducibleVulnerability}
	\end{subfigure}
	\caption{\label{fig:reducedVulnerabilityWhole}Illustrations for \autoref{exa:reducible-vulnerability}.}
\end{figure}

Effectively, each dependency reason in the tainted information flow constitutes a potential vulnerability that calls for further investigation.
But the tainted information flow may still contain \emph{false positives}, i.e. issues that do not turn out to be actual vulnerabilities.
Therefore, a manual review of the tainted information flow by the security analyst (step 5 in \autoref{fig:overview}) as described in the following is necessary.

The analyst reviews each dependency reason between source and sink rules contained in the tainted information flow against the access control policy of the API.
The review has one of the following two outcomes.
First, if the policy does not properly cover the tainted flow expressed by the dependency reason, it represents an \emph{unsecured tainted flow}; the policy thus needs to be adapted such that proper access control for this flow can be implemented; in this case, the dynamic analysis (see \autoref{subsec:analysis_dynamic}) can be used to prove that the dependency reason leads to a true positive, i.e., a BAC vulnerability, by demonstrating the occurrence of unsecured access to sensitive data.
Alternatively, the dependency reason represents a \emph{secured tainted flow}; in this case, the dynamic analysis can be used to verify the correct implementation of the policy.

We formalize these two alternative outcomes below.

\begin{defi}[secured and unsecured tainted flow]\label{def:review}
	Given a tainted type graph $TG_t = (TG, T)$, a tainted Graph API $(TG_t,\mathcal{R},src,sink)$, a tainted flow
	$f = \cup \{DR(r_{src},r_{sink})| r_{src} \in src(s), r_{sink} \in sink(s), s\in T\}$
	and an access control policy oracle $ACP = (RO, \leq_{RO}, P)$,
	then $f$ can be divided in two subsets:
	\begin{itemize}
		\item
		$S_{f}$ is the set of \emph{secured tainted flows} and consists of all dependency reasons in $f$ for which the security analyst has confirmed the correctness of the policy description based on their manual review.
		\item
		$U_{f} = f \setminus S_f$ the set of \emph{unsecured tainted flows} consisting of all dependency reasons for which the security analyst did not confirm the correctness of the policy description based on their manual review.
	\end{itemize}
\end{defi}
\noindent 
Recall the policy requirement stated in \autoref{subsec:preliminaries_access-control} based on the running example: \emph{a user may only update the information of a repository, if that user is the owner of the repository}; additionally, recall the tainted flow of the running example containing the pair \emph{(createRepo, updateRepo)}. Following a manual review, the security analyst can identify that the pair is covered by the policy. Therefore, in this case, $f = S_f = \{\emph{(createRepo, updateRepo)}\}$.
\subsection{Dynamic analysis}\label{subsec:analysis_dynamic}
The static analysis detects all potential BAC vulnerabilities, which may contain false positives, i.e., dependencies which do not constitute vulnerabilities. Therefore, the static analysis is not \emph{complete}; to achieve completeness, we complement the static analysis with the dynamic analysis. The dynamic analysis conducts actual API executions which validate potential BAC vulnerabilities detected by the static analysis; the dynamic analysis is complete, in that, all vulnerabilities it detects are indeed BAC instances. We describe the dynamic analysis segment of our approach below.

The dynamic analysis relies on the secured flows $S_f$ and unsecured flows $U_f$ identified during the static analysis---step 5 in \autoref{fig:overview} (left) and \autoref{def:review}. Drawing on $S_f$ and $U_f$, the security analyst defines test cases (step 6 in \autoref{fig:overview}) which:
\begin{itemize}
	\item If a dependency reason belongs to $S_f$, test whether the policy has been implemented correctly for this dependency reason.
	\item If a dependency reason belongs to $U_f$, expose the BAC vulnerability stemming from this dependency reason or discover an implicit undocumented implementation of the policy.
\end{itemize}
\noindent 
We refer to these test cases as \emph{taint tests}, and formally define them below. We distinguish positive taint tests from negative taint tests. For negative tests we expect a BAC exception to be raised, and for positive tests we expect no exception to be raised. If this is not the case, we have detected a BAC vulnerability.
\begin{defi}[taint test]\label{def:taint-test}
	Given a tainted Graph API $(TG_t,\mathcal{R},src,sink)$, with \emph{tainted flow}
	$f = \cup \{DR(r_{src},r_{sink})| r_{src} \in src(s), r_{sink} \in sink(s), s\in T\}$
	and access control policy oracle $ACP = (RO, \leq_{RO}, P)$,
	a \emph{taint test} $((t, roles_t), access)$
	consists of a role-based Graph API execution $(t, roles_t)$ and an expected outcome represented by the boolean variable $access \in \{true,false\}$.

	A taint test $((t, roles_t),access)$ is \emph{positive} (\emph{negative}) in case $access$ equals true (false).
\end{defi}
\noindent 
For a systematic method to the definition of test cases, we define the \emph{flow coverage}.
This coverage criterion helps with ascertaining whether taint tests covering the tainted flow are available. Flow coverage draws from similar approaches for covering rule dependencies~\cite{RungeKH13,HildebrandtLG13} of testing.\footnote{We focus in this first approach on \emph{covering dependency reasons} for pairs of sink and source rules, instead of covering dependencies between all rule pairs, or even dependency paths in a dependency graph.}
\begin{defi}[flow coverage]\label{def:flow-coverage}
	Given a tainted Graph API $(TG_t,\mathcal{R},src,sink)$, with \emph{tainted flow} $f = \cup \{DR(r_{src},r_{sink})| r_{src} \in src(s), r_{sink} \in sink(s), s\in T\}$
	and access control policy oracle $ACP = (RO, \leq_{RO}, P)$. A taint test $((t, roles_t),  access)$ \emph{covers} a dependency reason $d$ in $DR(r_{src},r_{sink})$ if for $(t : G_0 \stackrel{r_1}{\Rightarrow} G_1 \stackrel{r_2}{\Rightarrow} G_2 \stackrel{r_3}{\Rightarrow} G_3 \ldots, roles_t)$ there exists some $i,j > 0$ with $i<j$ such that $r_i = r_{src}$ and $r_j = {r_{sink}}$, and $t$ creates and uses elements as described in the dependency reason $d$.

	A set of taint tests $T$ satisfies the \emph{(unsecured, or secured) flow coverage} if for each dependency reason in $f$ ($U_f$, or $S_f$, resp.) there exists a positive and negative test in $T$.
\end{defi}

We call a taint test covering a dependency reason \emph{minimal}, when it consists of the source rule application directly followed by the sink rule application.

In order to make sure that we test with diverse roles from the access control policy, we formulate an additional coverage criterion. This criterion ensures that, for each role in the access control policy, there exists a positive and negative test such that: (i) both tests include two different API calls (ii) for positive tests, the second call is executed with a more privileged role than the first call or a role with the same privileges; for negative tests;  the second call is executed with a less privileged role than the first call.

\begin{defi}[role coverage]\label{def:role-coverage}
	Given a tainted Graph API $(TG_t,\mathcal{R},src,sink)$, with \emph{tainted information flow} $f = \cup \{DR(r_{src},r_{sink})| r_{src} \in src(s), r_{sink} \in sink(s), s\in T\}$ and access control policy oracle $ACP = (RO, \leq_{RO}, P)$.
	A taint test $((t, roles_t),  access)$ \emph{covers} a pair $(ro,ro')$ from $RO \times RO$ if for $(t : G_0 \stackrel{r_1}{\Rightarrow} G_1 \stackrel{r_2}{\Rightarrow} G_2 \stackrel{r_3}{\Rightarrow} G_3 \ldots, roles_t)$ there exists some $i,j \geq 0$ with $i < j$ such that $ro \in roles_t(i)$ and $ro' \in roles_t(j)$.

	Let $ro,ro'$ be elements from $RO$. A set of taint tests $T$ satisfies \emph{role coverage} if for each role $ro$ in $RO$ (1) there exists a positive test in $T$ covering some $(ro,ro')$ from $\leq_{RO}$ and (2) there exists a negative test in $T$ covering $(ro,ro')$ with $(ro,ro')$ from $>_{RO}$ except if $ro$ is the least privileged role in $ACP$.
\end{defi}

Based on the flow and role coverage, a security analyst is in position to generate a set of concrete test cases for a \gapi implementation. Their execution (step 7 in \autoref{fig:overview}) comprises the execution of API calls, as prescribed by the graph transformation sequence captured in a taint test. We have realized steps 6 and 7 by capturing these concrete test cases as unit-tests in a script which evaluates tests automatically---see \autoref{sec:graphql}.
During evaluation, we observe if an expected exception is included in the API response\eg in GitHub such an exception indicates that the call is attempting to perform an action for which the user executing the call has insufficient privileges.

Regarding the test case evaluation results and their correspondence to taint tests, we distinguish the following cases:
\begin{itemize}
	\item
	A \emph{positive} taint test $((t, roles_t),true)$ is \emph{successful} if its execution does \emph{not} lead to a BAC exception.
	\item A \emph{positive} taint test $((t, roles_t),true)$ \emph{fails} if its execution leads to a BAC exception.
	\item A \emph{negative} taint test $((t, roles_t),false)$ is \emph{successful} if its execution leads to a BAC exception.
	\item A \emph{negative} taint test $((t, roles_t),false)$ \emph{fails} if its execution does \emph{not} lead to a BAC exception.
\end{itemize}
\noindent 
The dynamic analysis is \emph{complete}, in that, if a positive (or negative) taint test fails, then a BAC vulnerability has been detected. Depending on the context in which the security analysis of the \gapi is performed, if the dynamic analysis focuses on flows from $S_f$, the security analyst should expect to obtain successful positive and negative tests---a test evaluation that would not meet these expectations would indicate an incorrect policy implementation; if the dynamic analysis focuses on flows in $U_f$, the security analyst should expect to obtain failing positive or negative tests.

As an example, recall the (secured) tainted flow obtained for the running example: $f=\{$\emph{(createRepo, updateRepo)}$\}$. A security analyst could satisfy the secured flow coverage by performing a positive and a negative test. The positive test should entail the execution of the query \emph{createRepo} (see \autoref{fig:schema}) as a user with an owner role and subsequently the execution of \textit{updateRepo} as a user with the same role. The negative test should entail the execution of \emph{createRepo} as a user with an owner role and subsequently the execution of \textit{updateRepo} as a user with the collaborator role\ie a user who is not an owner. As in this case the flow is secured, the expected outcome is that the positive test should succeed\ie the test should lead to no BAC exception and the repository information should be updated; the negative test should also succeed\ie the test should lead to a BAC exception and the repository information should not be updated. Assume a hypothetical scenario in which the security analyst concludes that the tainted flow $f=\{$\emph{(createRepo, updateRepo)}$\}$ is not secured properly and is thus part of the unsecured tainted flow. This may be the case because no information on securing the call \emph{updateRepo} is found in the policy description. Then, the expected outcome for the negative test would be that it fails\ie does not lead to a BAC exception, because indeed no proper access control for this case is implemented. On the other hand, this may be the case if the policy appears to be too severe and thus the expected outcome for the positive test would be that it fails\ie does lead to a BAC exception, although it should not.

The running example considers two roles: the owner and the collaborator. Regarding the role coverage, the tests described for the flow coverage above, cover the owner role: they include (i) a positive test with two calls executed by a user with an owner role (ii) a negative test with two calls, the first one executed by a user with an owner role and the second one executed by a less privileged user\ie the user with a collaborator role. To satisfy the role coverage, the security analyst could add one more test: a positive test with two calls executed by a user with a collaborator role. The reason the positive test suffices is that the collaborator is the least privileged role, so no negative test using this role can be defined.
\section{Application of the Taint Analysis to GitHub}\label{sec:graphql}
We applied the presented taint analysis to the GitHub API. In particular, we focused on the free version of GitHub and repositories that are owned by personal accounts.

We performed the taint analysis for two examples: the first application of the analysis builds on an \emph{extended version of the running example} (see \autoref{subsec:preliminaries_graphAPIs}), and aims to test a basic part of the access permissions of personal accounts~\cite{GitHub__AccesspermissionsGitHub}; the second application reproduces a \emph{real issue with permissions} which has been reported in the GitHub community forum, and therefore includes more details from the GitHub API. In each application, we perform the three segments of our approach---see \autoref{fig:activity-diagram}. We precede the presentation of the applications by relevant information on the GitHub access policy and a summary of the application setting. Following the presentation of the applications, we discuss decisions we made during the application and their implications. All artifacts mentioned in this section are contained in the \emph{digital package} accompanying the article~\cite{Khakharova_2025_Supplement}. The package also contains instructions on reproducing the applications.
\subsection{GitHub Access Control Policy}\label{subsec:graphql_github}
Each user of GitHub has a personal account. A repository owned by a personal account has two available roles or, as they are referred to in the GitHub documentation, \emph{permission levels}: the repository \emph{owner} and \emph{collaborators}. Owners have full control over a repository\eg they can delete the repository, whereas collaborators are typically only allowed to read and write to the code in a repository. Roles can only be assigned via the GitHub website and not via the API; the GitHub Access Policy does not allow administrative actions---see Section 2.2. When using the GitHub API, user authentication proceeds via \emph{tokens} associated with the personal accounts of users. While there are two types of personal access tokens, the \emph{classic} and the \emph{fine-grained}, in the remainder, unless otherwise indicated, we focus on the classic token---henceforth referred to simply as token. The permissions of a token correspond to the role of the token owner, but can be further limited by so-called \emph{scopes} granted to a token. For example: the scope \emph{repo} grants a collaborator full access to a repository, including read and write privileges as well as access to deployment statuses; the scope \emph{repo\_deployment} grants access to the deployment statuses; a token of a collaborator associated only with the \emph{repo\_deployment} scope grants the collaborator access to the deployment statuses but no access to the code.

The GitHub API access control policy requires that a user interacting with the API is authenticated using their token. The API allows the execution of calls only if the executing user has a role with sufficient permissions and the used token is associated with the necessary scopes; otherwise the API returns an error. The GitHub policy consists of relevant sections in GitHub's online documentation which map roles and scopes to possible actions and, thus, implicitly to API calls. As an example, recall the collaborator from above whose token is only associated with the \emph{repo\_deployment}; the documentation mentions that this scope only allows access to the deployment statuses without allowing access to the code; in terms of API calls, this means that the token only enables the collaborator to fetch information about the field \texttt{deployments} from the object \texttt{Repository} in the GitHub GraphQL schema, but not to access other unrelated fields.%
\subsection{Application Setting}\label{subsec:graphql-analysis_setting}
The presented applications focused on the examples and omitted other data on GitHub. In technical terms, the applications assume an empty initial graph (see \autoref{def:graphAPI}), although this is not required by the presented taint analysis; the analysis only requires that the initial graph does not contain any tainted nodes.

The analysis focused on three roles: the \emph{Owner}, who has the owner permission level and the \emph{repo} scope---giving users with this role full control over the repository plus some additional privileges\eg deleting the repository; the \emph{Collaborator}, who has the collaborator permission level and the \emph{repo} scope---giving users with this role full control over the repository minus the additional privileges of the owner; and the \emph{NoPe-Collaborator}, who has the collaborator permission level but no scopes---granting users with this role read-only access to public information about the repository\eg user profile info. In technical terms, we relied on a role specification where the set of roles $RO$ from our formal definition of a \gapi in \autoref{def:role-based-graphAPI} is equal to $\{\text{\emph{Owner, Collaborator, NoPe-Collaborator}}\}$, while the partial order $\leq_{RO}$ is equal to $\{$\emph{(NoPe-Collaborator, NoPe-Collaborator), (NoPe-Collaborator, Collaborator), (Collaborator, Collaborator), (Collaborator, Owner), (Owner, Owner)}$\}$.

The static analysis segment (see \autoref{fig:overview}) is performed based on the \emph{Henshin project}. Specifically, the GraphQL schemata were manually translated into class diagrams\ie Ecore files, whereas the rules were captured in Henshin diagrams. The CPA (coarse) was then executed using the relevant Henshin feature, presented in~\cite{Born_2017_GranularityConflictsDependenciesGraphTransformationSystems}.

The dynamic analysis segment entails the creation of Python scripts with positive and negative test cases for the tainted flows obtained by the static analysis. Each test case entails a role-based API execution (see \autoref{def:role-based-graphAPI})\ie a sequence of API calls realized by GraphQL queries and mutations and executed using a given token. The scripts use three tokens, corresponding to the roles in $RO$ from above. Positive tests entail the application of source and sink rules by roles with sufficient permissions, whereas negative tests entail the application of a source rule by a role with sufficient permissions and the application of the sink rule by a role with insufficient permissions. The test cases are \emph{minimal}\ie an application of a source rule is directly followed by an application of a sink rule; if the application of a sink rule requires a certain setup\eg the creation of a project requires a repository to be present, this setup is created prior to the sink rule.

The scripts put test cases inside unit tests. The unit-tests are executed and evaluated based on the Pytest library~\cite{Krekel__pytesthelpsyouwritebetterprogramspytestdocumentation}. The execution of the API calls in the unit-tests is performed using the Python GraphQL client~\cite{Krinke__pythongraphqlclientPythonGraphQLClient}. The scripts create the necessary resources ad-hoc; after the execution of the unit tests, the scripts delete all created resources, thus returning again to the empty initial graph state from above.
\subsection{Application to (Extended) Running Example}\label{subsec:graphql_running-example}
\begin{figure}[t]
	\centering
	\includegraphics[scale=.35]{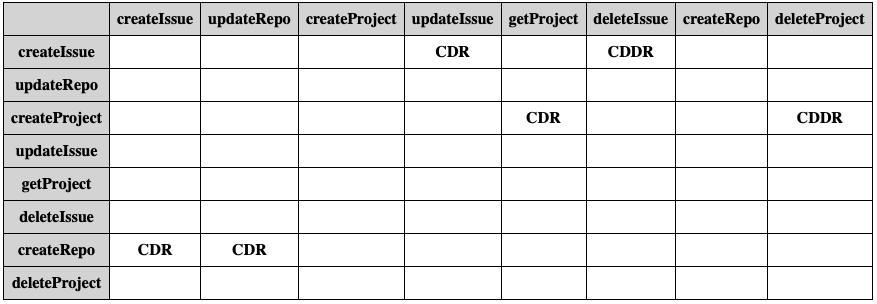}
	\caption{\label{subfig:henshin-output} The output of the dependency analysis (coarse) in Henshin for all the rules in \autoref{fig:rules}; CDR stands for \emph{Create Dependency Reason} and CDDR for \emph{Create-delete Dependency Reason}.}
\end{figure}
We commenced with the setup segment of the analysis and completed it as follows. First, we extended the running example: as presented in \autoref{sec:preliminaries}, in the tainted type graph $TG_t$ only a single node was tainted\ie the \texttt{Repository}, a single source rule was considered\ie \emph{createRepo}, and a single sink rule\ie \emph{updateRepo}. For the extended example, since the unauthorized manipulation of projects and issues could also compromise the data security of the platform, we also tainted the nodes \texttt{Project} and \texttt{Issue} in the schema in \autoref{fig:typeGraph}; moreover, according to queries and mutations in the GitHub API schema, we defined source and sink rules for the creation, deletion, and updating of instances of these newly tainted nodes. These new rules, including the previously presented \emph{createRepo} and \emph{updateRepo} rules, are illustrated in \autoref{fig:rules}.

We continued with the static analysis segment. As discussed earlier, we use Henshin for tool support of the static analysis.  An illustration of the output by Henshin\ie the dependency reasons, of the CPA is shown in \autoref{subfig:henshin-output}. Overall, the CPA returned the following dependencies\ie tainted information flow: (\emph{createIssue}, \emph{updateIssue}), (\emph{createIssue}, \emph{deleteIssue}), (\emph{createProject}, \emph{getProject}), (\emph{createRepo}, \emph{updateRepo}), (\emph{createProject}, \emph{deleteProject}), (\emph{createRepo}, \emph{createIssue}).
Besides the pair (\emph{createProject}, \emph{deleteProject}), for which we found no explicit documentation, all the other detected dependencies are addressed by the documentation for repositories, issues, and projects---see~\cite{GitHub__AccesspermissionsGitHub}. Therefore, the pair in question belongs to the unsecured tainted flow, whereas all other pairs belong to the secured flow---see \autoref{def:review}.

For the dynamic analysis segment, we created a script containing positive and negative taint tests based on the tainted flows obtained by the static analysis. In terms of \autoref{def:taint-test}, a taint test consists of an API execution $(t, role_t)$ and its expected outcome \emph{access}. For example, to address the dependency \emph{(createIssue, updateIssue)}, the script contains a (positive) taint test consisting of the API execution $t: G_0 \stackrel{r_1}{\Rightarrow} G_1 \stackrel{r_2}{\Rightarrow} G_2$  with $r_1$ the \emph{createIssue} rule, $r_2$ the \emph{updateIssue} rule, $role_t(0)=\text{\emph{Owner}}$,  $role_t(1)=\text{\emph{Collaborator}}$, and $\text{\emph{access}}=\text{\emph{true}}$; the script furthermore contains a (negative) taint test consisting of the API execution $t: G_0 \stackrel{r_1}{\Rightarrow} G_1 \stackrel{r_2}{\Rightarrow} G_2$  with $r_1$ the \emph{createIssue} rule, $r_2$ the \emph{updateIssue} rule, $role_t(0)=\text{\emph{Owner}}$, $role_t(1)=\text{\emph{NoPe-Collaborator}}$, and $\text{\emph{access}}=\text{\emph{false}}$.

The set of test cases in the script satisfies (secured and unsecured) \emph{flow coverage} (see \autoref{def:flow-coverage}), as it contains twelve test cases\ie a positive and a negative test case per dependency. Moreover, in order to satisfy the \emph{role coverage} (see \autoref{def:role-coverage}), the set contains six more test cases which ensure the criteria of the role coverage are met\eg in the \emph{(createIssue, updateIssue)} test pair from above, the set adds another taint test, consisting of the API execution $t: G_0 \stackrel{r_1}{\Rightarrow} G_1 \stackrel{r_2}{\Rightarrow} G_2$  with $r_1$ the \emph{createIssue} rule, $r_2$ the \emph{updateIssue} rule, $role_t(1)=\text{\emph{Owner}}$,  $role_t(2)=\text{\emph{Owner}}$, and $\text{\emph{access}}=\text{\emph{true}}$; thereby, this last test ensures that role coverage is satisfied for the owner role in $RO$ in the taint test in question.

For dependencies in the secured flow, all positive and negative tests were successful\ie positive tests did not lead to a BAC exception whereas negative tests led to a BAC exception---see \autoref{subsec:analysis_dynamic}. The same was true for the dependency in the unsecured flow which indicates that, although undocumented, access control is implicitly implemented for this dependency by the GitHub API.
\subsection{Application to GitHub Permissions Issue}\label{subsec:graphql_real-issue}
\begin{figure}[t]
	\centering
	\includegraphics[scale=.5]{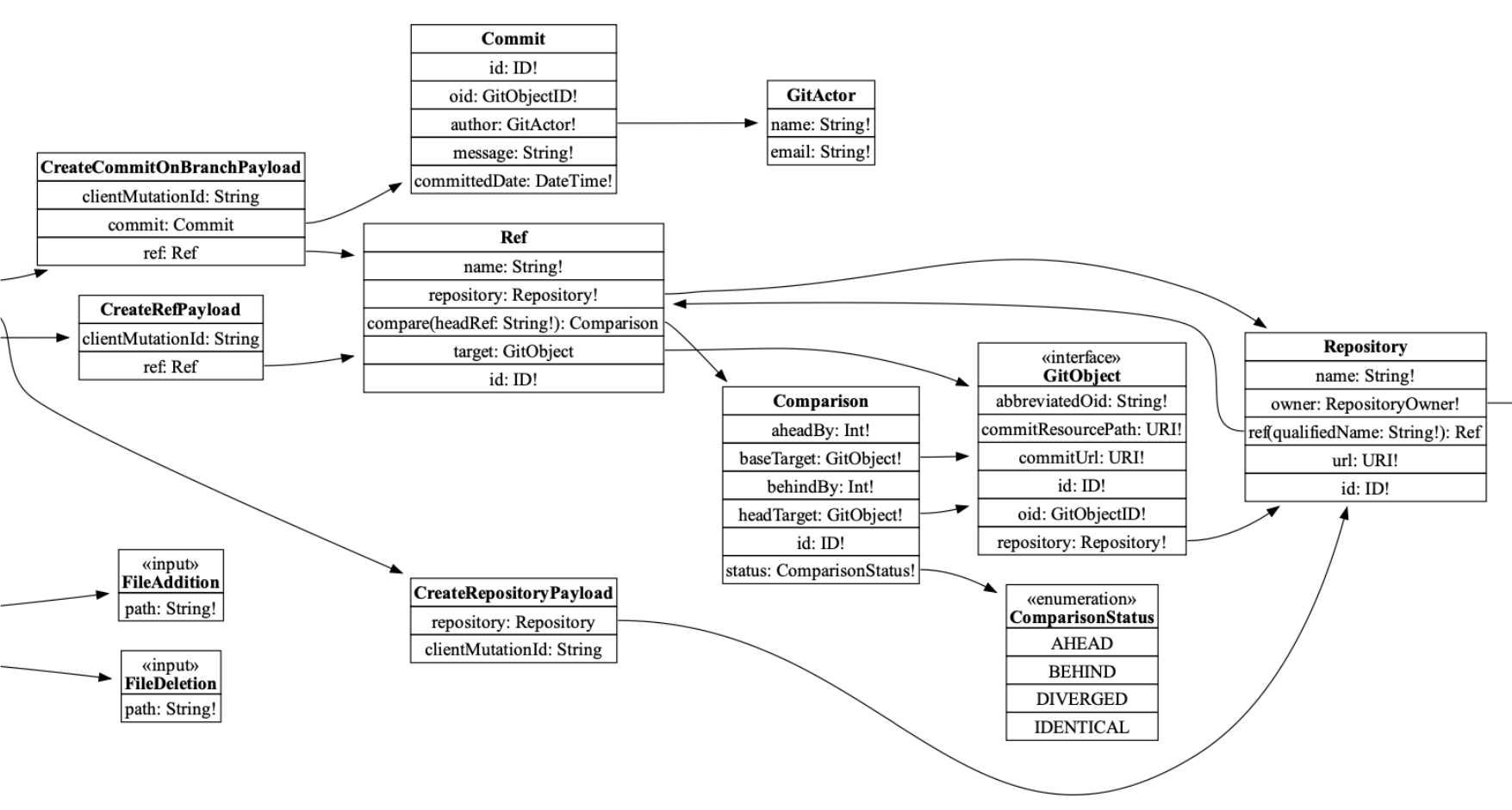}
	\caption{\label{fig:schema-real-excerpt} A partial graphical representation of the GraphQL schema used for the taint analysis of the GitHub permission issue in \autoref{subsec:graphql_real-issue}. The complete representation can be found in the digital package accompanying the article~\cite{Khakharova_2025_Supplement}.}
\end{figure}%
To demonstrate that the proposed taint analysis can be used to test a real vulnerability, we also applied the analysis to an issue with permissions reported on the GitHub community forum---see issue 106598 in~\cite{GitHub__Discussions}. The issue is related to the field \texttt{compare} in the object type \texttt{Ref}; the \texttt{Ref} object represents a branch in Git, and the \texttt{compare} field compares the latest commit in the branch represented by a given \texttt{Ref} object with the latest commit in another user-provided branch---see \autoref{fig:schema-real-excerpt}. The issue is that executing the GraphQL API call that resolves this field\ie performs the comparison, only requires that a token has the \emph{repo} scope when the user is authenticated using the OAuth method; if a user is authenticated using a fine-grained token (see \autoref{subsec:graphql_github}), then the scope is not required. Our goal was to apply the analysis to obtain taint tests which would perform the relevant actions\ie creation of a \texttt{Ref} object and execution of a comparison, and subsequently reproduce the vulnerability by differentiating the authentication method used while executing the API calls in the tests.

\begin{figure}
	\begin{subfigure}[t]{\textwidth}
		\centering
		\includegraphics[width=.9\textwidth]{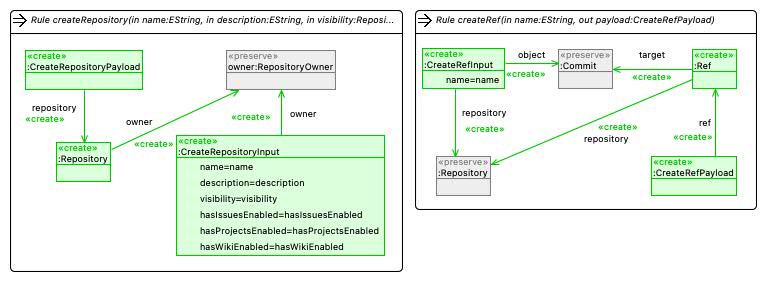}
		\caption{\label{subfig:createRepo-rules-real-example}The rules for creating a \texttt{Repository} (left) and a \texttt{Ref}.}
	\end{subfigure}
	\begin{subfigure}[t]{\textwidth}
		\centering
		\includegraphics[width=.8\textwidth]{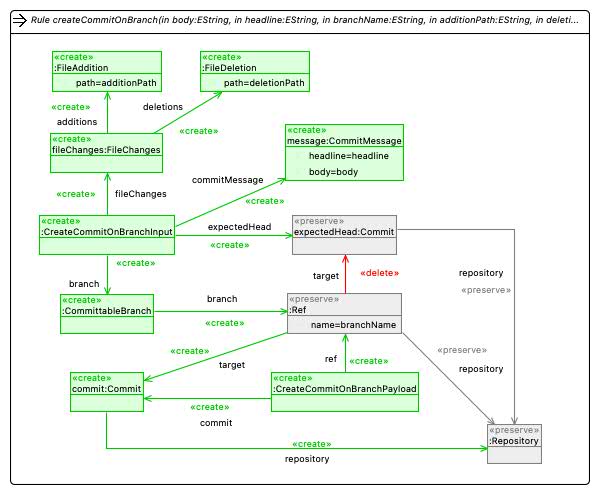}
		\caption{\label{subfig:createCommit-rules-real-example}The rule for creating a \texttt{Commit}.}
	\end{subfigure}
	\begin{subfigure}[t]{\textwidth}
		\centering
		\includegraphics[width=.9\textwidth]{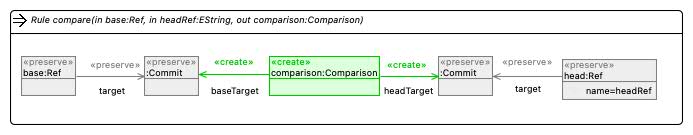}
		\caption{\label{subfig:compare-rules-real-example} The rule for performing a \texttt{Comparison}.}
	\end{subfigure}
	\caption{\label{fig:rules-real} The source and sink rules used in the permissions issue example in \autoref{subsec:graphql_real-issue}.}
\end{figure}

The tainted type graph $TG_t$ obtained in the setup segment addressed more aspects of the GitHub schema and adhered to the conventions of the GraphQL schema language. The $TG_t$ contains fields with input parameters\eg the \texttt{compare} field from above---see \autoref{fig:schema-real-excerpt}, that may alter the behavior of the field resolution; in UML terms, these fields resemble operations and have been modeled as such in the $TG_t$. The GraphQL schema language requires that object types are not used as input parameters; instead, for providing complex, non-scalar input to fields, the schema language defines \emph{input types}\ie data structures, which may include object types. Our $TG_t$ for this example fulfills these requirements. As in the GraphQL schema language, the $TG_t$ contains the two special types \texttt{Query and Mutation}; the content of these types are operations corresponding to queries and mutations in the GraphQL schema. Regarding the output of queries and mutations, the GitHub schema supports complex data structures by automatically generating \emph{payload} objects which comprise all output values from a query or a mutation. Our $TG_t$ correspondingly introduces payload classes for the output of queries and mutations. The GitHub schema defines enums and custom scalar\ie primitive, data types, which are modeled in the $TG_t$ via enum and data type classes, respectively. Finally, the GitHub schema equips many objects with a unique global ID, which can be used to directly retrieve an object and thus to improve performance\eg when caching. In the $TG_t$, the ID fields have been replaced with a many-to-one reference.

Regarding the set of rules, we focused on the mutations in the GitHub schema that create objects that are relevant to the example\ie mutations for creating a repository, a branch, and a commit. Finally, as mentioned above, the behavior of fields with input parameters resembles operations, therefore we also created a mutation for \texttt{compare} which, as in the GitHub schema, takes as parameter the branch to be compared to the \texttt{Ref} object from which the \texttt{compare} is resolved, and returns a \texttt{Comparison} object. The rules are illustrated in \autoref{fig:rules-real}. With respect to modeling of the rules in Henshin, we assume that the input for creating objects is provided as input parameters to the rule; however to represent the execution of a call to the GitHub API accurately, the rules entail the creation of input type objects which use the rule input parameters as fields---as in the GitHub schema. See for example the class \texttt{CreateRepositoryInput} in \autoref{subfig:createRepo-rules-real-example}. We tainted the nodes in the $TG_t$ that are relevant to the actions required for a comparison as well as the comparison itself\ie the nodes \texttt{Repository}, \texttt{Ref}, \texttt{Commit}, and \texttt{Comparison}.

\begin{figure}[t]
	\centering
	\includegraphics[width=.6\textwidth]{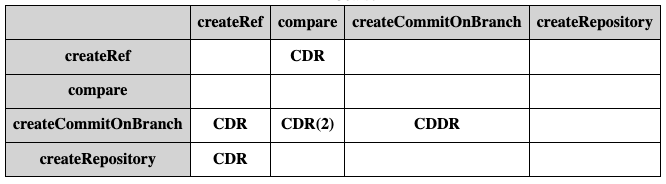}
	\caption{\label{fig:henshin-output-real} The output of the dependency analysis for the rules in \autoref{fig:rules-real}.}
\end{figure}
We completed the static segment by performing CPA between source and sink rules. The result of the analysis is shown in \autoref{fig:henshin-output-real}. The completion of the dynamic segment entailed the creation of a script which creates the necessary setup\ie the repository and two branches, and focuses on the dependency relevant to the reported issue\ie the creation of a commit followed by an execution of a comparison. The script executes two unit-tests: one is creating two commits and then performs the comparison using the OAuth authentication without the \emph{repo} scope, whereas the other creates two commits and uses the fine-grained token, again without the \emph{repo} scope. The execution of the script confirmed the observation reported in the issue 106598 in the GitHub community forum\ie the test using OAuth fails to execute the comparison, whereas the test using the fine-grained token succeeds.
\subsection{Discussion}\label{subsec:graphql_discussion}
The goal of the extended running example in \autoref{subsec:graphql_running-example} was to demonstrate the applicability of the approach as well as to exemplify the formal results\eg flow and role coverage. The goal of the reproduction of the permissions issue in \autoref{subsec:graphql_real-issue} was to obtain a \gapi which covered a larger part of the GraphQL  schema language, thereby addressing the translation of more GraphQL concepts into a graph-transformation-based formalism.

Our modeling decisions when translating the GraphQL schema into the artifacts required by the taint analysis\ie the type graph and the set of transformation rules, were based on two considerations. First, the translation  concerns transferring concepts from an implementation-oriented language into a language of higher abstraction. Therefore, our translation omitted some object types that are closely related to the implementation of applications that would use the API\eg types intended for facilitating the display of query results by the front-end, as well as concepts which are closely related to the runtime behavior\eg \emph{directives}, that is, annotations that enable the description of alternate runtime execution and type validation behavior in GraphQL. On the other hand, to comply to the GraphQL instructions and GitHub conventions, we modeled input types and payloads. The second consideration was that the translation occurs between a textual language for the description of structures into a graphical language with the same purpose. This consideration motivated our choice to translate fields referring to unique IDs into many-to-one references.

Obtaining the type graph and rules for the permissions issue also paved the way for automating this translation, which constitutes part of our future plans for the taint analysis (see \autoref{fig:overview-vision} and \autoref{sec:conclusion}). Although the translation for the permissions issue example addresses a large part of the GraphQL schema language, there are other parts which may require formalization extensions\eg the translation of fields with input parameters which lead to the creation of objects with other input parameters; in practice, such nested fields resemble the invocation of operations within operations and may thus require the nesting of rules or the introduction of control flow elements. We plan to realize such extensions by employing relevant Henshin features\eg units.
\section{Related Work}\label{sec:related}
To the best of our knowledge, there currently does not exist any approach to security testing for Graph APIs with simultaneous support for key features of our approach\ie using taint analysis, focusing on broken access control, relying on graphs, and allowing for the sound detection of vulnerabilities.
In the following, we first discuss three works that are closely related to our approach, in that, they support at least three of the key features.
We then discuss works that focus on either implementing or testing broken access control and rely on graph-based modeling.
Subsequently, we discuss other approaches to security testing for access control.
We then proceed with a discussion of static and dynamic taint analysis for security testing of mobile and web applications, and conclude the section with a discussion of testing approaches for GraphQL APIs.

The work in \cite{Mouelhi_2008_ModelBasedFrameworkSecurityPolicySpecificationDeploymentTesting} presents a framework for the specification, deployment, and testing of security policies, based on Model-driven Engineering.
The framework relies on a class diagram for representing the system structure and on a novel security metamodel, which may be seen as an informal representation of a typed graph, for enabling the specification of rule-based access control policies.
The metamodel enables the specification of policies based on several methods\eg role-based or organization-based.
The framework supports the automated transformation of such policies into security components.
In terms of testing, the authors present a fault model defined at the meta-level that is used for systematically generating test cases using mutation.
Moreover, the authors argue that these policies enable intrinsic consistency (via the metamodel-model relationship) as well as manually created conformance checks. In comparison, our work focuses solely on testing for role-based policies.
Moreover, our work deliberately assumes the existence of an \emph{informal} policy document and instead requires that the security analyst models API calls\ie the manners in which to create and manipulate data in the system, as graph transformation rules.
These rules enable a sound and systematic detection of dependencies between rules.
Dependencies are then verified by the creation of tests which can be adjusted to satisfy coverage criteria.
On the other hand, the generation of tests in our approach is currently manual.
The automation of this process, as well as the extension of the presented formalization to support other access policy methods, are planned in our future work.

Another model-based approach is presented in~\cite{Daoudagh_2020_XACMETXACMLTestingModeling}.
This work focuses on testing access control policies given in the \emph{eXtensible Access Control Markup Language} (XACML)~\cite{OASIS_2025_eXtensibleAccessControlMarkupLanguageXACMLVersion30}.
The approach represents such a policy as a tree that is then used to generate a typed graph\ie a model, representing the policy evaluation.
Based on this typed graph, the approach systematically generates an oracle as well as tests and can, moreover, assess the coverage of a given test suite.
While our approach also relies on an oracle being available, as mentioned in the comparison with the previous work, our work deliberately assumes that the policy is given in natural language. This decision differentiates the two approaches significantly:
our approach enables a static analysis, similar to the work in~\cite{Daoudagh_2020_XACMETXACMLTestingModeling}, via the modeling of the API calls as graph transformation rules but, moreover, entails the dynamic analysis to verify whether the policy has been implemented correctly.

The work in \cite{Mammar_2016_FormalDevelopmentSecureAccessControlFilter} uses UML notations to specify access control\ie the structure of the data is modeled as a class diagram while security rules are modeled as SecureUML and  activity diagrams.
The authors present rules that translate the UML notations into a specification in a formal language.
The specification can then be systematically checked for consistency and validated.
Subsequently, the validated specification can be used for the systematic derivation of a formal, proved access control filter. While this approach is not intended for testing but rather for implementing access control, this approach relates to our work in its modeling of the data as a graph\ie a class diagram, as well as in its capability for formal precision.

There exist a number of approaches formalizing, analyzing, and implementing access control policies based on \emph{graph-based representations}. Implementation of access control based on graphs typically entails the modeling of a policy as a graph that is stored in a graph database; the access control mechanism then checks compliance to the policy based on graph queries\eg~\cite{Yang_2024_GraphBasedFrameworkABACPolicyEnforcementAnalysis,PinaRos_2012_GraphbasedXACMLevaluation,Jin__XACMLImplementationBasedGraphDatabase}. Other works implement mechanisms to control access \emph{to} a graph database by rewriting queries such that are compliant to the policy before passing them on for evaluation to the database\eg~\cite{Hofer_2023_RewritingGraphDBQueriesEnforceAttributeBasedAccessControl,BereksiReguig_2024_EffectiveAttributeBasedAccessControlModelNeo4j}.

Regarding analysis and formalization, Koch et al.~\cite{KochMP05} use a sophisticated formalization to compare and analyze  different policy models. We decided to equip our taint analysis with a first relatively simple access control policy model. It lowers the barrier for applying our analysis and allow cases, where only natural language descriptions for access control policies are given. We plan to integrate more complex formalizations into our taint analysis in future work. Burger et al.~\cite{BurgerJW15} present an approach to detect (and correct) \emph{vulnerabilities in evolving design models} of software systems. Graph transformation is used to formalize patterns used to detect security flaws. This approach is not dedicated specifically to Graph APIs, nor broken access control.
The works in~\cite{RayLFK04} and~\cite{Alves_2017_graphbasedframeworkanalysisaccesscontrolpolicies} use graph-based approaches to visualize and analyze access control policies, while the work in~\cite{Mohamed_2021_ExtendedAuthorizationPolicyGraphStructuredData} extends XACML for graph-structured data.
The suitability of graph-based representations for representing policies and enabling access control is also being investigated by recent approaches that realize access control via Machine Learning~\cite{You_2023_knowledgegraphempoweredonlinelearningframeworkaccesscontroldecisionmaking,Yin_2024_heterogeneousgraphbasedsemisupervisedlearningframeworkaccesscontroldecisionmaking}.

A number of approaches to testing for broken access control rely on \emph{model-based testing}. These approaches typically generate tests for the access control based on a formal or semi-formal model of the policy or the system. For instance, the works in~\cite{Masood_2009_ScalableEffectiveTestGenerationRoleBasedAccessControlSystems, Fan_2024_TestCaseGenerationAccessControlBasedUMLActivityDiagram,Pretschner_2008_ModelBasedTestsAccessControlPolicies} model the policy as a finite state model, as activity diagrams, and as a metamodel, respectively; the work in~\cite{Xu_2012_modelbasedapproachautomatedtestingaccesscontrolpolicies} models the system as a high-level Petri Net. In comparison, our approach opts for assuming the policy is only informally documented and focuses on dependencies between interactions that may cause vulnerabilities. Moreover, our approach entails the dynamic analysis that aims to verify whether the access control has been implemented correctly.

Several \emph{static and dynamic taint analysis} approaches for security testing of \emph{mobile or web applications} exist, see e.g.~\cite{EnckGCCJMS14,KuznetsovGTGZ15,TrippPFSW09}. These approaches also analyze the data flow of security-relevant or private data, but there is currently no approach dedicated to Graph APIs, nor broken access control in particular. Pagey et al.~\cite{Pagey+23} present a security evaluation framework focusing also on access control vulnerabilities, but it is specifically dedicated to software-as-a-service e-commerce platforms.

The following approaches are dedicated to \emph{testing GraphQL APIs}, but do not focus on finding \emph{security vulnerabilities} such as broken access control.  Belhadi et al.~\cite{BelhadiZA24} present a fuzzing approach based on evolutionary search. Karlsson et al.~\cite{KarlssonCS21} introduce a property-based testing approach. Vargas et al.~\cite{vargas_deviation_2018} present a regression testing approach. Zetterlund et al.~\cite{ZetterlundTMB22} present an approach to mine GraphQL queries from production with the aim of testing the GraphQL schema.
\section{Conclusion and Future Work}\label{sec:conclusion}
\begin{figure}[t]
	\centering
	\begin{subfigure}[t]{.49\textwidth}
		\includegraphics[width=\textwidth]{figures/overview}
	\end{subfigure}
	\begin{subfigure}[t]{.49\textwidth}
		\includegraphics[width=\textwidth]{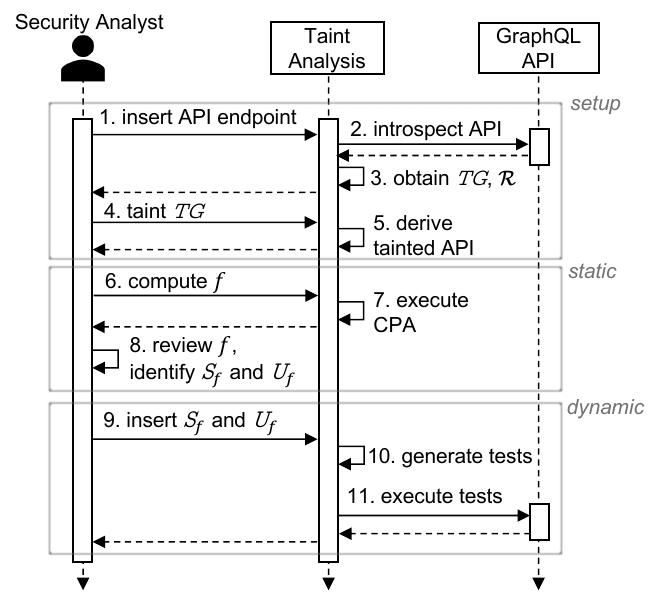}
	\end{subfigure}
	\caption{\label{fig:overview-vision} The overview of the sequence of steps in the presented approach from \autoref{fig:overview} (left), and the vision for the automation of the approach (right).}
\end{figure}
We have presented a \emph{systematic approach supporting security analysts to test for broken access control vulnerabilities in Graph APIs}.  The security testing approach is based on taint analysis and consists of a static and dynamic segment.  The static taint analysis primarily aims at validating the access control policies.  The dynamic taint analysis aims at finding errors in the implementation of the (validated) access control policies.

Future work comprises \emph{increasing the level of automation} of the taint analysis as illustrated in \autoref{fig:overview-vision} (right). In particular, it is part of current work to provide automatic translations from Graph APIs such as GraphQL into tainted typed graph transformation systems.  Naturally, the validation then still will include a manual review, but the formal specification could at least in parts be mined from the GraphQL implementation. There is also room for improving automation of the dynamic segment of our analysis. Currently, we provided automated test execution scripts, but we are also planning to explore the automation of the test generation step or coverage analysis for existing test suites. Moreover, we want to \emph{evaluate the approach} on further and larger case studies and test its generality, also in the context of other relevant domains such as access control to graph databases. As mentioned in \autoref{sec:analysis}, we plan to devise a \emph{more fine-grained} sound detection technique for potential \emph{indirect vulnerabilities} related to broken access control relying on a new type of critical pair analysis for indirect dependencies. Moreover, as mentioned in Section~\ref{sec:related}, a further line of future work is the integration of more \emph{detailed graph-based formalizations} of (other types of) access control. Finally, although for simplicity we omitted attribute manipulation and concepts like inheritance in the current formalization of Graph APIs, it seems feasible to incorporate such more sophisticated concepts based on the state-of-the-art for graph transformation.
\section*{Acknowledgement}
\noindent The authors would like to thank the anonymous reviewers for their insightful comments and suggestions.
\bibliographystyle{alphaurl}
\bibliography{literature}
\end{document}